\pdfoutput=1

\documentclass[10pt,twocolumn,twoside]{IEEEtran}

\usepackage{subcaption}
\usepackage{chngcntr} 

\newcommand\mypara{\par}

\counterwithin*{para}{section} 

\usepackage[utf8]{inputenc}
\usepackage[english]{babel}
\usepackage{mathrsfs} 

\usepackage{enumerate}

\newcommand{\tb}[1]{#1}
\newcommand{\tbb}[1]{#1}

\ifCLASSINFOpdf
  \usepackage{xcolor}
  \usepackage[pdftex]{graphicx}
  \graphicspath{{./imgs/}}
  \DeclareGraphicsExtensions{.pdf,.jpeg,.png}

  \usepackage{epstopdf} 
\else
\fi
\usepackage{amsmath, amsthm, amssymb, amsfonts}
\theoremstyle{plain}
\newtheorem{theorem}{Theorem}

\newtheorem{corollary}{Corollary}

\theoremstyle{definition}
\newtheorem{definition}{Definition}

\newtheorem{problem}{Problem}

\theoremstyle{remark}
\newtheorem{remark}{Remark}

\DeclareMathOperator{\diag}{diag}

\usepackage{array}

\usepackage{url}

\usepackage{algorithm}
\usepackage{algorithmic}

\begin{document}
\title{Distributed Nonlinear Control Design using Separable Control Contraction Metrics}

\author{Humberto~Stein Shiromoto, Max Revay, and 
        Ian~R. Manchester
\thanks{The authors are with the Australian Centre for Field Robotics, The University of Sydney, NSW 2006, Australia. Corresponding author: {\tt ian.manchester@sydney.edu.au}}
}

%
%

%

\maketitle

\begin{abstract}
This paper gives convex conditions for synthesis of a distributed  control system for large-scale networked nonlinear dynamic systems.  It is shown that the technique of control contraction metrics (CCMs) can be extended to this problem by utilizing {\em separable} metric structures, resulting in controllers that only depend on information from local sensors and communications from immediate neighbours. The conditions given are pointwise linear matrix inequalities, and are necessary and sufficient for linear positive systems and certain monotone nonlinear systems. Distributed synthesis methods for systems on chordal graphs are also proposed based on SDP decompositions. \tb{The results are illustrated on a problem of vehicle platooning with heterogeneous vehicles, and a network  of nonlinear dynamic systems with over 1000 states that is not feedback linearizable and has an uncontrollable linearization}.
\end{abstract}


%
\IEEEpeerreviewmaketitle

\section{Introduction}

\mypara In recent years, rapid advances in communication and  computation technology have enabled the development of large-scale engineered systems such as smart grids \cite{hill_smart_2012}, sensor networks \cite{Pajic2011}, smart manufacturing plants \cite{wang_implementing_2016}, and intelligent transportation networks \cite{CanudasdeWitMorbidiLeonOjedaEtAl2015}. Despite these advances, the systematic design of feedback controllers for such large systems remains challenging.

\mypara When it is assumed that a system  has linear dynamics and that all sensor information can be collected in a single location for control computation, well-developed synthesis methods such as LQG and $H^\infty$ can be applied \cite{anderson1990optimal, dullerud2000course}. However, emerging applications motivate going beyond these assumptions.

Firstly, for geographically distributed systems with hundreds or thousands of nodes, such as transportation and power networks, it is not practical to collect all sensor information in one location for control. In this case there is a need for {\em distributed} methods that rely only on information available locally or communicated from nearby nodes.

Secondly, most real systems exhibit {\em nonlinear} dynamics. When large excursions in operating conditions are expected, e.g. due to changing production demands in a flexible manufacturing system, or recovery from a fault in a smart electrical grid, one must take into account the system nonlinearity.


\mypara Decentralised and distributed control are long-standing problems in control theory, with important early work detailed in \cite{sandell_survey_1978} and \cite{siljak1978large}. A key concept is the \textit{vector Lyapunov function}, i.e. a Lyapunov function made up of individual storage functions for the nodes, a concept closely related to the separability property we use in this paper.  Terminology is not completely uniform in the literature, but
in this paper we take ``decentralised'' to mean that at each node the controller uses {\em only} local state information, and  ``distributed'' to mean that {\em some} communication is allowed between nearby nodes. 

%

For linear state feedback, information flow can be encoded by a sparsity structure on the feedback gain matrix, however in general this problem can be NP-hard \cite{BlondelTsitsiklis1997}. It has been recognized by many authors that if the search is restricted to {\em diagonal} (or {\em block diagonal}) Lyapunov matrices, then the problem is convex (see, e.g., \cite{zecevic_control_2010, Tanaka2011, Rantzer2015} and references therein). The main benefit is that sparsity structure in the gain matrix is preserved under the standard change of variables for LMI-based design. In general, restricting the set of Lyapunov functions is conservative: it produces sufficient conditions for stabilizability, but not necessary conditions. However, for the important sub-class of systems for which internal states are always  non-negative, known as {\em positive systems}, existence of a diagonal Lyapunov function is actually necessary and sufficient
(see, e.g., \cite{berman_nonnegative_1994} and references therein). This result has been extended to $H^\infty$ design \cite{Tanaka2011}, and scalable algorithms for control design \cite{Rantzer2015} and identification \cite{umenberger_scalable_2016} of networked positive systems.

\mypara Design of controllers for nonlinear systems has also been a major topic of research for many years, see e.g. \cite{Slotine1991,Krstic1995, Khalil:2001} for established approaches.
Most methods require (at least implicitly) the construction of a \emph{control Lyapunov function}. While for certain structured systems, constructive methods such as backstepping and energy-based control can be used \cite{Krstic1995}, no general methodology exists. Indeed, the set of control Lyapunov functions can be non-convex and disconnected \cite{Rantzer:2001}, which poses a challenge for synthesis. 


\mypara A drawback of standard Lyapunov functions is the fact that they are defined with respect to a particular set-point or limit set, which must be known {\em a priori}. When the target trajectory may change in real time, a common situation in robotics or flexible manufacturing, it is more appropriate to define a function depending on the distance {\em between} pairs of points. Tools such as contraction metrics \cite{Lohmiller1998} and incremental Lyapunov functions \cite{Angeli2002} provide such a capability for stability analysis. Contraction concepts have proven useful in the analysis of networked systems, in particular oscillation synchronization and entrainment \cite{Wang2004, PhamSlotine2007, russo2010global, Aminzare2014a}, and techniques for contraction analysis based on sum-separability properties of metrics \cite{Russo2013, como2015throughput, coogan2016separability, Manchester2017existence}. Extensions to reaction-diffusion PDE systems have appeared in \cite{arcak2011certifying}, where again a metric is constructed that integrates over space, generalizing the notion of sum-separability to continuous spaces.

The concept of a control contraction metric (CCM) was introduced in \cite{Manchester2014a, manchester_control_2017} and extends contraction analysis to constructive control design. The main advantages this method offers over the Lyapunov approach are that the synthesis conditions are convex, and it provides a stabilizing controller for {\em all} forward-complete solutions, not just a single set-point. It was shown in \cite{manchester_control_2017} that the CCM conditions are necessary and sufficient for feedback-linearizable nonlinear systems.



\tb{
\mypara The main contributions of this paper are the following.
\begin{enumerate}
	\item We extend the results of \cite{manchester_control_2017} to show that by imposing a \textit{separable} structure on a control contraction metric, a distributed nonlinear feedback controller can be obtained via convex optimization, with the property that all \textit{on-line} computations can be performed with prescribed information sharing between nodes.
	\item We provide necessary conditions for the existence of a separable metric for certain classes of monotone systems.
	\item We show that the \textit{off-line} convex search for a CCM can scale to large-scale systems with chordal graph interaction structure.\end{enumerate}}
The  conference paper \cite{SteinShiromotoManchester2016} presented preliminary results related to, but less general than, the results of this paper. In particular, it considered completely decentralized design, and did not address scalability of the resulting computations. The main result of \cite{SteinShiromotoManchester2016} is Corollary 1 in this paper. 


\section{Preliminaries and Problem Formulation}

\subsection{Notation} We use the notation $\mathbb{R}_{\geq0}$ for the non-negative reals, and $\mathbb{N}_{[a,b]}$ with $a<b$ for natural numbers between $a$ and $b$.
Let $n>0$ be any integer, the vector $e_i$ denotes the vector with zeros in all \tb{entries } except the $i$-th where it is 1. Given $N$ matrices $M_{1},\ldots,M_{N} \in\mathbb{R}^{p\times q}$, the notation $\diag(M_{1},\ldots,M_{N})$ denotes the block matrix $M\in\mathbb{R}^{Np\times Nq}$ with the $M_i$ matrices on the main (block) diagonal, and zeros elsewhere. The notation $M\succ 0$ (resp. $M\succeq 0$) stands for $M$ being positive (resp. semi)-definite. The sets of (semi)-definite symmetric matrices are denoted as $\mathbf{S}_{\star0}^n=\{M\in\mathbb{R}^{n\times n}:M\star0,M=M^T\}$, where $\star\in\{\succ,\succeq,\prec,\preceq\}$. 

\mypara The notation $\mathcal{L}_{\mathrm{loc}}^\infty(\mathbb{R}_{\geq0},\mathbb{R}^m)$ stands for the class of functions $u:\mathbb{R}_{\geq0}\to\mathbb{R}^m$ that are locally essentially bounded. Given differentiable functions $M:\mathbb{R}^n\to\mathbb{R}^{n\times n}$ and $f:\mathbb{R}^n\to\mathbb{R}^n$ the notation $\partial_fM$ stands for matrix with dimension $n\times n$ and with $(i,j)$ element given by $\frac{\partial m_{ij}}{\partial x}(x)f(x)$. The notation $\dot{f}$ always stands for the total derivative with respect to time $t$.

\mypara Let $N>0$ be an integer, a \emph{graph} consists of a set of \emph{nodes} $\mathscr{V}\subset\mathbb{N}_{[1,N]}$ and a set of \emph{edges} $\mathscr{E}\subset\mathscr{V}\times\mathscr{V}$ and it is denoted by the pair $(\mathscr{V},\mathscr{E})=\mathscr{G}$. A node $i\in\mathscr{V}$ is said to be \emph{adjacent} to a node $j\in\mathscr{V}$ if $(i,j)\in\mathscr{E}$, the set of nodes that are adjacent to $j$ is defined as $\mathscr{N}(j)=\{i\in\mathscr{V}:i\neq j,(i,j)\in\mathscr{E}\}$. A graph is said to be \emph{undirected} if, for every edge $(i,j)\in\mathscr{E}$, there exists $(j,i)\in\mathscr{E}$. It is said to be \emph{directed} if otherwise. \tb{For a directed graph $\mathscr{G}=(\mathscr{V},\mathscr{E})$, we define an undirected graph $\mathscr{G}^u=(\mathscr{V},\mathscr{E}^u)$ with  $(i,j)\in \mathscr{E}^u$ (and hence also $(j,i)\in \mathscr{E}^u$) if either $(i,j)\in \mathscr{E}$ or $(j,i)\in \mathscr{E}$, or both. Given two graphs with the same vertex set $\mathscr{G}_1 = (\mathscr{V}, \mathscr{E}_1)$, $\mathscr{G}_2 = (\mathscr{V}, \mathscr{E}_2)$, we define their union $\mathscr{G}_1 \cup \mathscr{G}_2$ to be the graph $(\mathscr{V}, \mathscr{E}_1 \cup \mathscr{E}_2 )$, i.e. the graph that contains all edges appearing in either graph.}

Given two nodes $i,j\in\mathscr{V}$, an ordered sequence of vertices $v_k, k = 1, ..., n$ with $v_1=i, v_n=j$ and  $(v_k,v_{k+1})\in\mathscr{E} \,\forall k$ is said to be a \emph{path from node $i$ to the node $j$}. A path is said to be \emph{cycle} if node $i$ equals node $j$, no edges are repeated, and the nodes $i$ and $j-1$  are distinct.

\tb{For an undirected graph, the following concepts are recalled from \cite{PakazadHanssonAndersenEtAl2015,VandenbergheAndersen2015}.  A graph is said to be a \emph{tree} if it is connected and does not contain cycles. 
 A \emph{clique} $\mathscr{C}\subset\mathscr{V}$ of the graph $\mathscr{G}$ is a maximal set of nodes that induces a complete (fully connected) subgraph on $\mathscr{G}$. A \emph{chord} of a cycle is any edge joining two nonconsecutive nodes.  A graph is said to be \emph{chordal} if every cycle of length greater than three has a chord. The importance of a graph being chordal is that it has a tree-decomposition into cliques  \cite[Proposition~12.3.11]{Diestel2005} such a tree is said to be a \emph{clique tree} and it is denoted as $\mathscr{T}(\mathscr{G})$.}

\subsection{Networked System Definition}\label{sec:network}

\tb{In this paper, we consider systems made up of a network of $N\in\mathbb{N}$ nodes. Interconnection between the nodes is defined by two directed graphs: a physical interaction network graph $\mathscr{G}_p$ and a communication network graph $\mathscr{G}_c$. Both graphs have the same vertex set  $\mathscr{V}=\mathbb{N}_{[1,N]}$ corresponding to system nodes, but may have different edge sets, as illustrated in Figure~\ref{fig:graph different illustration}. We assume both graphs have self-loops at each node, i.e. $(i,i)$ is in the edge set for all $i\in\mathscr{V}$.

\begin{figure}
	\centering
	\includegraphics[width=\columnwidth]{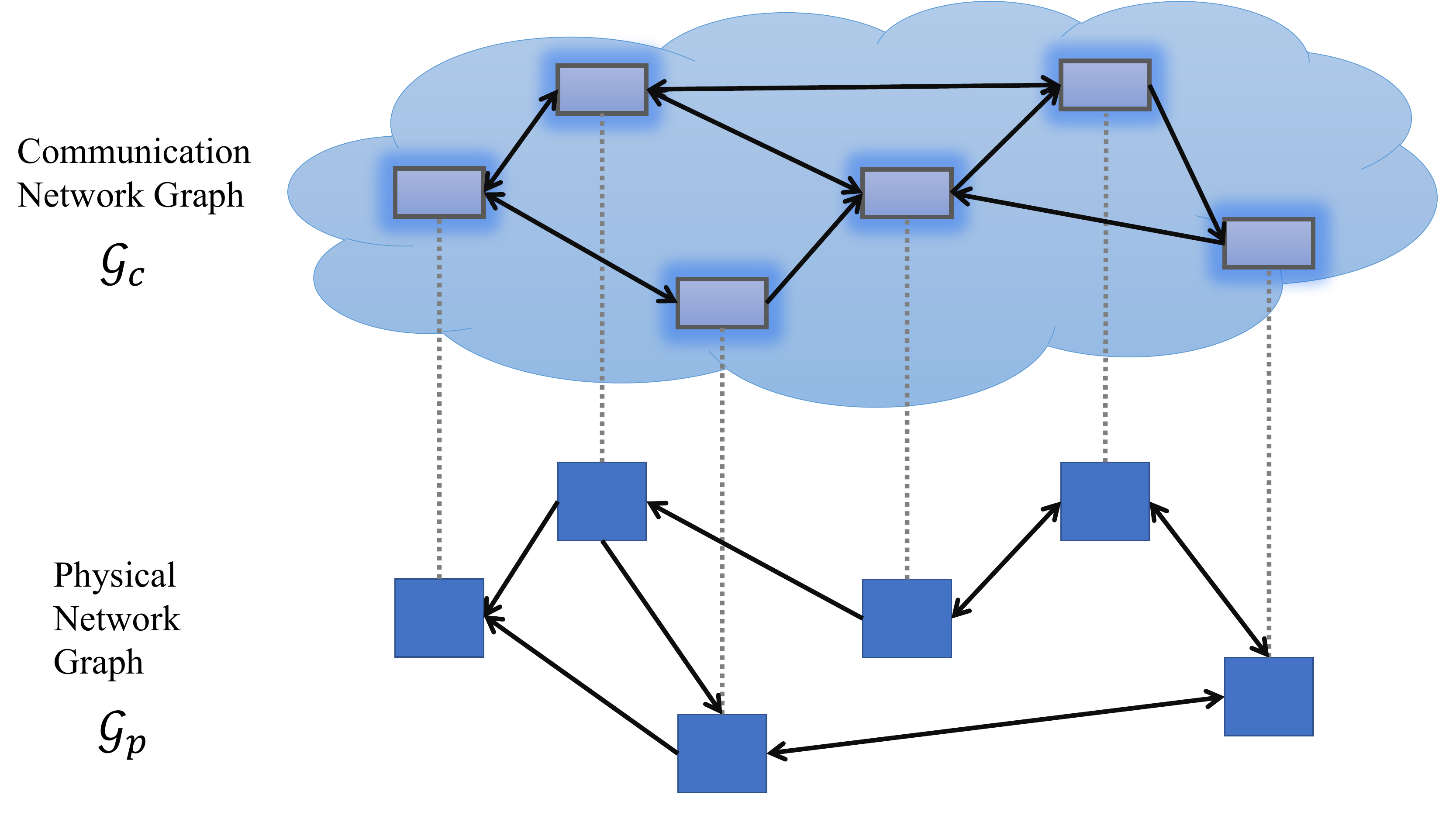}
	\caption{Illustration of the directed graphs representing the physical interaction between nodes $\mathscr{G}_p$, and the communication network $\mathscr{G}_c$. These may or may not be identical.}
	\label{fig:graph different illustration}
\end{figure}

The physical graph $\mathscr{G}_p=(\mathscr{V}, \mathscr{E}_p)$  defines the dynamical interaction between individual nodes.
At each node $i\in \mathscr{V}$, there is local state vector $x_i\in\mathbb{R}^{n_i}$ and  control input $u_i\in\mathbb{R}^{m_i}$. We define $\breve{x}_i\in\mathbb{R}^{\breve{n}_i}$ as a stacked vector of node states $x_j$ 
for which $(j,i)\in\mathscr{E}_p$, i.e. all nodes that influence $x_i$. Each node's dynamics are governed by the differential equation:
\begin{subequations}\label{eq:general system}
\begin{equation}\label{eq:subsystem}
	\dot{x}_i(t)=f_i(x_i(t),\breve{x}_i(t))+b_i(x_i(t), \breve{x}_i(t))u_i(t),
\end{equation}
We allow the case that for some nodes $i\in\mathscr{V}$, $b_i(x_i, \breve{x}_i)=0$ and $m_i=0$, i.e. node $i$ has no direct control input.
For the complete networked system we will also use the notation
\begin{equation}\label{eq:stacked system}
	\dot{x}(t)=f(x(t))+B(x(t))u(t), 	
\end{equation}
\end{subequations}
with stacked vectors and functions
\[
x = \begin{bmatrix}
	x_1\\ \vdots \\ x_N
\end{bmatrix}\in\mathbb{R}^n, ~
u = \begin{bmatrix}
	u_1\\ \vdots \\ u_N
\end{bmatrix}\in\mathbb{R}^m,  ~
f = \begin{bmatrix}
	f_1\\ \vdots \\ f_N
\end{bmatrix},
\] 
and input matrix $B = \diag(b_1, \ldots, b_N)$.
The functions $f:\mathbb{R}^n\to\mathbb{R}^n$ and $B:\mathbb{R}^n\to\mathbb{R}^{m\times n}$ are assumed to be smooth, i.e., infinitely differentiable.

}

Similarly, the graph $\mathscr{G}_c=(\mathbb{N}_{[1,N]},\mathscr{E}_c)$ specifies a communication network, in that $(j,i)\in\mathscr{E}_c$ if node $j$ can send instantaneous measurements of its state to node $i$ for control computation, and $\vec{x}_i\in\mathbb{R}^{\vec{n}_i}$ is a stacked vector of node states $x_j\in\mathbb{R}^{n_j}$ such that  $(j,i)\in\mathscr{E}_c$.

\subsection{Universal Exponential Stabilizability}\label{sec:Problem Formulation and Motivation}

\mypara A function $u^\star\in\mathcal{L}_{\mathrm{loc}}^\infty(\mathbb{R}_{\geq0},\mathbb{R}^m)$ is said to be an \emph{input signal or control for \eqref{eq:general system}}. For such a control for \eqref{eq:general system}, and for every \emph{initial condition} $x^\star(0)$, there exists a unique solution to \eqref{eq:general system} (\cite{Teschl2012}) that is denoted by $X(t,x^\star(0),u^\star)$, when evaluated at time $t$. This solution is defined over an open interval $(\underline{t},\overline{t})$, and it is said to be \emph{forward complete} if $\overline{t}=+\infty$.  \tb{We define a \textit{target trajectory} as a pair $(x^\star, u^\star)$ where $x^\star=X(\cdot,x^\star(0),u^\star)$ is a forward-complete solution of \eqref{eq:general system}. \tbb{Given a communication graph $\mathscr{G}_c$ we define $\vec{x}_i^\star$ analogously to $\vec{x}_i$ above.}

Following  \cite{Manchester2014a, manchester_control_2017}, the system \eqref{eq:general system} is said to be \textit{universally exponentially stabilizable} with rate \tbb{$\lambda>0$ } if there exists a feedback controller \tbb{$k:\mathbb R^n\times \mathbb R^n \times  \mathbb R^m  \to \mathbb R^m$ }and  a constant value $C>0$
such that for \textit{every} target trajectory $(x^\star, u^\star)$, solutions $x(t)$ of the closed-loop system
\[
\dot x(t) = f(x(t))+B(x(t))k(x(t),x^\star(t), u^\star(t)))
\]
exist for all $t\ge 0$ and satisfy 
\begin{equation}\label{eq:global exponential uniform stabilizability}
	\left|x^\star(t)-x(t)\right|\leq Ce^{-\lambda t}|x^\star(0)-x(0)|
\end{equation}
for every initial condition $x(0)\in\mathbb{R}^n$. Note that this is a stronger condition than global exponential stabilizability of a \textit{particular} target trajectory, such as the origin.

\subsection{Problem Statement}

The main objective of this paper is to find a distributed controller that can stabilize \textit{any} trajectory of a particular system. To formalize this we make the following definition:

\begin{definition}
	A state feedback controller $u(t)=k(x(t),t)$ for the system \eqref{eq:general system} is said to be to be \textit{ $\mathscr{G}_c$-admissable} if it decomposes into $N$ local feedback laws of the form 
\[
u_i(t)=k_i(x_i,\vec{x}_i,x_i^\star,\vec{x}_i^\star,u_i^\star),
\] 
for $i\in\mathbb{N}_{[1,N]}$. That is, each local control signal depends only on local state and target trajectory information and neighbor information communicated in accordance with $\mathscr{G}^c$.
\end{definition}

We are now ready to state formally the distributed control problem we consider in this paper.

\begin{problem}\label{problem formulation}
For the system \eqref{eq:general system}, find a $\mathscr{G}_c$-admissable state feedback controller such that for any target trajectory $(x^\star, u^\star)$, the closed-loop system satisfies \eqref{eq:global exponential uniform stabilizability} for almost all $x(0)\in\mathbb{R}^n$.
\end{problem}

The ``almost all $x(0)\in\mathbb{R}^n$'' condition simplifies the resulting CCM control construction, however the result can be extended to ``all $x(0)\in\mathbb{R}^n$'' by the sampled-data controller constructed in \cite{manchester_control_2017}.

}


\subsection{Differential Dynamics and Control Contraction Metrics}\label{sec:background}

\mypara We  recall some standard facts from Riemannian geometry (see e.g. \cite{Boothby1986} for a complete development). A \emph{Riemannian metric on $\mathbb{R}^n$} is a symmetric positive-definite bilinear form that depends smoothly on $x\in\mathbb{R}^n$. In a particular coordinate system, for any pair of vectors $\delta_0,\delta_1$ of $\mathbb{R}^n$ the \emph{metric} is defined as the inner product $\langle\delta_0,\delta_{ 1}\rangle_x=\delta_0^T M(x)\delta_1$, where $M:\mathbb{R}^n\to\mathbb{R}^{n\times n}$ is a smooth function. Consequently, ``local'' notions of norm $|\delta_0|_x^2=\langle\delta_0,\delta_0\rangle_x=:V(x,\delta_0)$ and orthogonality $\langle\delta_0,\delta_1\rangle_x=0$ can be defined on the tangent space. The metric is said to be \emph{bounded} if there exists constants  $\underline{m}>0$ and $\overline{m}>0$ such that, for all $x\in\mathbb{R}^n$, $\underline{m}I_n\leq M(x)\leq \overline{m}I_n$, where $I_n\in\mathbb{R}^{n\times n}$ is the identity matrix. 

\mypara Let $\Gamma(x_0,x_1)$ be the set of piecewise-smooth curves $c:[0,1]\to\mathbb{R}^n$ connecting $x_0=c(0)$ to $x_1=c(1)$. The Riemannian \emph{energy} of $c$ is
\begin{align*}
	e(c)=&\ \int_0^1\left|c_s(s)\right|^2_{c(s)}\,ds = \int_0^1 V(c(s),c_s(s))\,ds\;,
\end{align*}
where the notation $c_s$ stands for the derivative $\tfrac{\partial c}{\partial s}$. The Riemannian energy between $x_0$ and $x_1$, denoted as $e(x_0,x_1)$, is defined as the minimal energy of a curve connecting them:
\begin{equation}\label{eq:geodesic formulation}
	e(x_0,x_1)=\inf_{c\in\Gamma(x_0,x_1)}e(c)\;.
\end{equation} 
This curve is smooth and is referred to as a {\em geodesic}.

\mypara Along each solution of \eqref{eq:general system}, one can define the \emph{differential} (a.k.a. variational or prolonged) dynamics:
\begin{equation}\label{eq:general system:differential}
	\dot{\delta}_x=A(x,u)\delta_x+B(x)\delta_u\;,
\end{equation}
where $\delta_x$ (resp. $\delta_u$) is a vector of the Euclidean space $\mathbb{R}^n$ (resp. $\mathbb{R}^m$) and the matrix $A\in\mathbb{R}^{n\times n}$ has components given by 
\begin{equation*}
	A_{jk}(x,u)=\dfrac{\partial}{\partial x_k}\left[f_j+\sum_{i=1}^{m} B_{ji}u_i\right]
\end{equation*}
for indices $j,k\in\mathbb{N}_{[1,n]}$. The differential dynamics \eqref{eq:general system:differential} describe the behaviour of tangent vectors to curves of solutions of \eqref{eq:general system}.


\mypara Similarly to \eqref{eq:general system}, given a control $\delta_u$ for system \eqref{eq:general system:differential}, the solution to \eqref{eq:general system:differential} computed at time $t\geq0$, along solutions $(x(t),u(t))$ of \eqref{eq:general system}, and issuing from the  initial condition $\delta_x\in\mathbb{R}^n$ is denoted by $\Delta_x(t,x(0),\delta_x(0),u,\delta_u)$. 



\mypara A sufficient condition for the stability of \eqref{eq:general system:differential} is provided by analyzing the derivative of a particular function along the solutions of systems \eqref{eq:general system} and \eqref{eq:general system:differential} \cite{Lohmiller1998}. 

\begin{definition}
		 A bounded metric $V:\mathbb{R}^n\times\mathbb{R}^n\to\mathbb{R}_{\geq0}$ is called a \emph{contraction metric} for \eqref{eq:general system} if, for any control $u$ for system~\eqref{eq:general system}, there exists a scalar $\lambda>0$ such that the inequality
	 	 \begin{equation}
	 	 	\tfrac{d}{dt}V(x(t),\delta_x(t))\leq -2\lambda V(x(t),\delta_x(t))
 	 	 \end{equation}
 	 	 holds, where $x(t):=X(t,x(0),u)$ and $\delta_x(t):=\Delta_x(t,x(0),\delta_x(0),u,0)$, for every pair $(x,\delta_x)\in\mathbb{R}^n\times\mathbb{R}^n$.
\end{definition}
In particular, a metric $V(x,\delta_x)=\delta_x^T M(x)\delta_x$ is a contraction metric for \eqref{eq:general system} if the following linear matrix inequality 
\begin{equation}\label{eq:contraction_M}
\dot M(x)+A(x,u)M(x)+M(x)A(x,u)\preceq -2\lambda M(x)
\end{equation}
holds for all $x, u$ \cite{Lohmiller1998}. Since $\dot M =\partial_{f(x)+B(x)u}M(x)$ and $A(x,u)$ are affine in each control input $u_i$, this implies that the corresponding coefficient matrices must be zero:
\begin{equation}\label{eq:killing}
		\partial_{b_i}{M(x)}+\tfrac{\partial b_i}{\partial x}^T M(x)+M(x)\tfrac{\partial b_i(x)}{\partial x}	=0
		\end{equation}
		for each $i \in\mathbb N_{[1,m]}$, which means $b_i$ are Killing vectors for the metric $M$. In that case, the inequality \eqref{eq:contraction_M} is equivalent to
\begin{equation}\label{eq:Artstein-Sontag:inequality}
			\partial_{f}{M(x)}+\tfrac{\partial f}{\partial x}^T M(x)+M(x)\tfrac{\partial f}{\partial x}\preceq -2\lambda M(x).
\end{equation}
In the remainder of the paper, we will often drop explicit dependence on $x$ of $M(x)$ and other matrices for brevity, but these matrices are always state dependent unless explicitly stated otherwise.

\mypara The existence a contraction metric for system \eqref{eq:general system} implies that every two solutions to this system converge to each other exponentially with rate $\lambda$. To the authors knowledge, this was first proven in \cite{Lewis1949} using Finsler metrics, a more general class than Riemannian metrics. \tb{The paper \cite{Forni2014} introduced the concept of a Finsler-Lyapunov function to further investigate relationships between Finsler structures and differential notions of stability and contraction.

Contraction analysis was extended to constructive control design in \cite{manchester_control_2017} by introducing the concept of a \textit{control contraction metric}.


\begin{definition}[{\cite{manchester_control_2017}}]
	A bounded metric is said to be a \emph{control contraction metric for system \eqref{eq:general system}} if \eqref{eq:killing} holds and there exists a constant value $\lambda>0$ such that for $\delta_x\neq 0$ we have the implication	
		\begin{equation}\label{eq:Arstein-Sontag}
			\delta_x^T MB=0 \Rightarrow \delta_x^T\left(\partial_{f}{M}+\tfrac{\partial f}{\partial x}^TM+M\tfrac{\partial f}{\partial x}+2\lambda M\right)\delta_x<0.
		\end{equation}
\end{definition}
\tbb{Condition \eqref{eq:Arstein-Sontag} } can be interpreted as the requirement that the system be contracting \tbb{in all } directions \textit{orthogonal} to the span of the control inputs. 
It was shown in \cite{manchester_control_2017} that this is equivalent to the existence of a differential feedback gain $\delta_u=K(x)\delta_x$ for which 
\begin{equation}\label{eq:CCM MK}
	\dot M +(A+BK)^TM+M(A+BK)+2\lambda M \prec 0
\end{equation}
for all $x, u$, which leads to the following control design method.

\paragraph*{Step 1} (Offline LMI computation) The inequality \eqref{eq:CCM MK} is equivalent (see \cite{manchester_control_2017})  to the existence of a bounded ``dual metric'' $W:\mathbb{R}^n\to\mathbb{R}^{n\times n}$ such that $W(\cdot)=W(\cdot)^T\succ0$ and a function $Y:\mathbb{R}^n\to\mathbb{R}^{m\times n}$ satisfying the following linear matrix inequality 
\begin{equation}\label{eq:MI formulation}
	-\dot{W}+AW+WA^T+BY+(BY)^T+2\lambda W\prec0
\end{equation}
for all $(x,u)\in\mathbb{R}^n\times\mathbb{R}^m$. Note that \eqref{eq:MI formulation} is linear in the matrix functions $W$ and $Y$. Consequently, for polynomial systems, the pointwise LMI \eqref{eq:MI formulation} can be solved via sum of squares programming \cite{Parrilo2003}. For non-polynomial systems, these constraints could be approximately satisfied either via polynomial approximation of dynamics, bounding of dynamics via linear differential inclusions \cite{boyd1994linear}, or via gridding the state/input space.

\mypara Once a solution to LMI~\eqref{eq:MI formulation} has been computed, the function defined, for every $(x,\delta_x)\in\mathbb{R}^n\times\mathbb{R}^n$, by
	\begin{equation}\label{eq:differential feedback law}
		\delta_u=Y(x)W^{-1}(x)\delta_x:=K(x)\delta_x
	\end{equation}
	is a differential feedback law that renders the origin globally exponentially stable for system~\eqref{eq:general system:differential} in closed loop with $\delta_u$.

\paragraph*{Step 2} (Online controller computation). The feedback law for system \eqref{eq:general system} can be obtained by integration as follows. 
\begin{enumerate}
	\item Compute a minimal geodesic:
\begin{equation}\label{eq:ccm geo min}
\gamma =\arg\min_{c\in\Gamma(x^\star(t),x(t))}e(c)\
\end{equation}
\item Integrate the differential controller
\begin{align}
	u(t)&=k(x(t),x^\star(t),u^\star(t))\notag\\&=u^\star(t)+\int_0^1 K(\gamma(t,s))\gamma_s(t,s)\,ds,\label{eq:contracting feedback law}
\end{align}
\end{enumerate}
For a bounded metric, the Hopf-Rinow theorem (cf. \cite[Theorem 7.7]{Boothby1986}) ensures that for every pair $x(t), x^\star(t)$, there exists a minimizing geodesic $\gamma$ solving \eqref{eq:ccm geo min}. Furthermore, for each $x^\star(t)$ this geodesic is unique and a smooth function of $x(t)$ for almost all $x(t)$.}

\begin{remark}
In the case that the metric $M=W^{-1}$ is independent of $x$, the unique minimal geodesic is a straight line joining $x$ to $x^\star$. Furthermore in the case that $Y$ and hence $K$ are also independent of $x$, the above controller reduces to a linear feedback law
\[
u(t)= k(x(t),x^\star(t),u^\star(t))= u^\star(t)+K(x(t)-x^\star(t)),
\]
so \eqref{eq:contracting feedback law} can be thought of as a natural generalisation of linear feedback synthesis to nonlinear systems.  
\end{remark}
\tb{
\begin{remark}\label{rem:convex}
	For Theorem \ref{thm:main result}, we have assumed that \eqref{eq:MI:SSCCM} holds for all $x\in\mathbb{R}^n$. If \eqref{eq:MI:SSCCM} holds only on a subset $S \subset \mathbb{R}^n$, then it is necessary to ensure that $\gamma(s)$ remains in this subset for all $s\in[0,1]$. This is the case if both $x$ and $x^\star$ are in $S$ for all $t$, and $S$ is geodesically convex. For constant metrics, geodesic convexity is the standard convexity in $\mathbb{R}^n$, since geodesics are straight lines.
\end{remark}}


\section{Convex Design of Distributed Controllers}\label{sec:Design of Decentralized Controllers}

\tb{In this section, we present the main results of the paper, extending the CCM methodology described above to distributed control design. } Inspired by the notion of sum-separable Lyapunov functions (see e.g. \cite{Dirr2015}), we introduce the following class of control contraction metrics:
\begin{definition}\label{def:SSCCM}
	A control contraction metric $V$ for system \eqref{eq:subsystem} is called \emph{sum-separable} if it can be decomposed like so:	\begin{equation*}
		V(x,\delta_x)= \sum_{i=1}^N V_i(x_i,\delta_{x_i}):=\sum_{i=1}^N \delta_{x_i}^T M_i(x_i)\delta_{x_i},
	\end{equation*}
	where, for each index $i\in\mathbb{N}_{[1,N]}$, and for every $(x_i,\delta_{x_i})\in\mathbb{R}^{n_i}\times\mathbb{R}^{n_i}$, the function $V_i(x_i,\delta_{x_i})$ is a metric on $\mathbb{R}^{n_i}$.
\end{definition}

\mypara In other words, Definition~\ref{def:SSCCM} states that the metric $V$ on $\mathbb{R}^n$ can be  decomposed into a sum of smaller components, each of which depends only on the {\em local} information $x_i, \delta_{x_i}$. Accordingly, we define the following class of matrix functions:

\tb{
\begin{definition}
For the system \eqref{eq:subsystem}, let $\Pi$ denote the set of matrix functions $\mathbb{R}^n\to\mathbf{S}_{>0}^n$ with the following properties:\begin{enumerate}
	\item Each $M(x)\in \Pi$ is block diagonal with $N$ blocks, and the $i^{th}$ block has dimension $n_i$.
	\item The $i^{th}$ block of $M(x)$ is a function only of $x_i$.
\end{enumerate}
\end{definition} 
}

I.e. a sum separable CCM $V(x)=\delta^T M(x)\delta$ has $M(x)\in \Pi$. Note that $M(x)\in\Pi \Leftrightarrow M(x)^{-1}\in \Pi$.

\mypara To address the information constraints on $k$ described in Problem~\ref{problem formulation}, the structure of the feedback defined by Equation~\eqref{eq:contracting feedback law} is obtained by imposing a suitable constraint on the function $Y$ to be satisfied together with the LMI \eqref{eq:MI formulation}. 

\tb{
\begin{definition}
Let $\Xi$ be the set of functions $Y:\mathbb{R}^n\times\mathbb{R}^m\to\mathbb{R}^{m\times n}$ with components defined by
\begin{equation*}
		Y_{ij}\begin{cases}
		=Y_{ij}(x_i,\vec{x}_i)\in\mathbb{R}^{m_i\times n_i},&\text{if}\ (i,j)\in\mathscr{E}_c,\\
		\equiv0_{m_i\times n_i},&\text{otherwise},
		\end{cases}
\end{equation*}
for every $i,j\in\mathscr{V}$.	
\end{definition} }

\mypara The set $\Xi$ defines the topology of the differential feedback law to be designed for system~\eqref{eq:general system:differential} and the dependence of each element of the matrix $Y$ on the state-space variables. 

\begin{theorem}\label{thm:main result}
	For the system \eqref{eq:general system} and differential dynamics \eqref{eq:general system:differential}, suppose there exist	$W(x)\in\Pi, \quad
			Y(x)\in\Xi$		
			satisfying the following pointwise linear matrix inequality:	
			\begin{equation}\label{eq:MI:SSCCM}
			-\dot{W}+AW+WA^T+BY+(BY)^T+2\lambda W\prec0		\end{equation}	
		for all $x\in\mathbb{R}^n, u\in\mathbb{R}^m$. 
	Then, $M(x)=W(x)^{-1}$ defines a separable control contraction metric for system~\eqref{eq:subsystem} and the controller \eqref{eq:contracting feedback law} with $K(x)=Y(x)W(x)^{-1}$ solves Problem \ref{problem formulation}.
	\end{theorem}

\begin{proof}
To prove the theorem we first establish $\mathscr{G}_c$-admissability of the controller, and then that it achieves the desired form of stability.

By assumption, $W\in\Pi$, so we also have $M=W^{-1}\in \Pi$, and therefore $M$ defines a sum-separable metric, as per Definition \ref{def:SSCCM}.

At a particular time $t$, the first stage of control calculation is to compute a minimum-energy geodesic from $x(t)$ to $x^\star(t)$. Because $M$ is sum-separable, the energy of any curve $c:[0,1]\to\mathbb{R}^n$ satisfies the following equation
	\begin{equation}\label{eq:geodesic sum}
		e(c)=\int_0^1\sum_{i=1}^N V_i\left(c_i(s),\frac{\partial c_i}{\partial s}(s)\right)\,ds\;.
	\end{equation}
where $c_i:[0,1]\to\mathbb{R}^{n_i}$ denotes the $i^{th}$ component of the curve $c$, connecting $x_i(t)$ to $x^\star_i(t)$. Defining the energy \tbb{of } each component $c_i$ as
\[
e(c_i) = \int_0^1V_i\left(c_i(s),\frac{\partial c_i}{\partial s}(s)\right)ds
\]
and exchanging the order of integration and summation we have $e(c) = \sum_{i=1}^N e_i(c_i)$. Hence computing the curve of minimal energy $e(c)$ decomposes into computing the component curves $c_i$ of minimal energy $e_i(c_i)$, each of which depends only on local information $x_i(t), x^\star_i(t)$.

Hence each local controller at node $i$, with knowledge of $x_i(t), \vec x_i(t), x_i^\star(t), \vec x_i^\star(t)$, can compute the minimal geodesics $\gamma_i(t)$ and $\vec \gamma_i(t)$, referring to the stacked vector function of geodesics $\gamma_j(t)$ for $j:(j,i)\in\mathscr{E}^c$.

The second stage of the control computation is integration of the differential control law. Since $M\in\Pi$, i.e. both block diagonal and with local state-dependence of the blocks, the transformation $K(x)=Y(x)W(x)^{-1}=Y(x)M(x)$ preserves the sparsity pattern and local dependence of $Y(x)$, so $K(x)\in\Xi$. This means that the $i,j$ block of $K(x)$ can be written as $K_{ij}(x_i, \vec x_i)$.

Then,  each local agent computes the control signal, where $t$-dependence of signals has been dropped for simplicity:
\begin{align}
	u_i=&u_i^\star+\sum_{j:(j,i)\in\mathscr{E}^c}\int_0^1 K_{ij}(\gamma_i(s), \vec\gamma_i(s))\frac{\partial\gamma_{j}}{\partial s}(s)\,ds.
\end{align}
By construction, this control signal satisfies $\mathscr{G}^c$-admissability.

%
%
%
%

%
	
\mypara	The LMI~\eqref{eq:MI:SSCCM} implies that the inequality
	\begin{equation*}
		\delta_x^T\bigg(\dot{M}+(A+BK)M+M(A+BK)^T+2\lambda M\bigg)\delta_x\leq0
	\end{equation*}
	holds, for every $(x,\delta_x,u)\in\mathbb{R}^n\times\mathbb{R}^n\times\mathbb{R}^m$. Thus, $M$ is a control contraction metric for system \eqref{eq:general system} and, according to the main result of \cite{manchester_control_2017}, \eqref{eq:differential feedback law} is a differential feedback rendering the equilibrium of the origin globally exponentially stable for system~\eqref{eq:general system:differential} in closed loop.
	\end{proof}

\begin{corollary}[{\cite{SteinShiromotoManchester2016}}]\label{cor:main result}
	Assume that the matrix $B$ satisfies the identity
	$
		\partial_BW-\tfrac{\partial B}{\partial x}W-W\tfrac{\partial B}{\partial x}^T\equiv0
$
	and there exist $N$ functions $\rho_i:\mathbb{R}^{n_i+\vec{n}_i}\to\mathbb{R}$ such that the matrix inequality
	\begin{equation}\label{eq:with R}
		-\partial_fW+\tfrac{\partial f}{\partial x}W+W\tfrac{\partial f}{\partial x}^T-BRB^T+2\lambda W\prec0
	\end{equation}
	holds for all $(x,u)\in\mathbb{R}^n\times\mathbb{R}^m$, where $R(x)=\mathbin{\diag}(\rho_1(x_1)I_{n_1},\ldots,\rho_N(x_N)I_{n_N})$ for some scalar functions $\rho_i(x_i), i = 1, .., N$. Then, $W$ is a sum-separable control contraction metric for system~\eqref{eq:general system} and there exists a solution to Problem~\ref{problem formulation} with fully decentralized information structure, i.e. $\mathscr{G}^c$ has no edges $(i,j), i\ne j$.
\end{corollary}

\mypara To see that Corollary \ref{cor:main result} is a particular case of Theorem~\ref{thm:main result}, note that by choosing $Y=-RB^T/2$, \eqref{eq:with R}  is equivalent to \eqref{eq:MI:SSCCM}. Furthermore, $Y$ by construction is block diagonal and the $i^{th}$ block depends only on $x_i$, hence $Y\in\Xi$. 


\begin{remark}
In the above we have assumed that each node consists of a node state $x_i$ and a colocated node control $u_i$. However, the above strategy is easily extended to a communication structure based on separate ``state measurement nodes'' $x_i$ and ``actuation nodes'' $u_j$, and a communication networks from sensors to actuators defined by a directed bipartite graph $\mathscr{G}_c$, the adjacency matrix of which defines the sparsity structure of $Y$. For the online control computation, at each measurement node the state $x_i(t)$ is measured, and a minimal geodesic path to $x_i^\star(t)$ is computed. Then this path is communicated to each control node $j$ such that $(i,j)$ is an edge of $\mathscr{G}_c$. Then each control node can compute the control according to \eqref{eq:contracting feedback law}.	
\end{remark}

\begin{remark}
As shown in \cite{manchester_control_2017}, the Riemannian energy function then provides a useful control-Lyapunov function for any target trajectory. In particular, at each time it defines a convex set of control signals that achieve exponential contraction towards the target trajectory. 
This was used in \cite{wang_distributed_2017} to guarantee stability in distributed economic model predictive control.
\end{remark}

\tb{
\subsection{Conditions for Existence of a Separable CCM}
The results we have presented so far give sufficient conditions for existence of a distributed controller by way of a separable CCM. A natural question to ask is how conservative is the restriction to a separable CCM.

\mypara For linear time-invariant \textit{positive} systems, i.e. those leaving the positive orthant invariant, stability is equivalent to the existence of a separable quadratic Lyapunov function \cite{berman_nonnegative_1994}. This leads to the following simple result:

\begin{theorem}\label{thm:local}
Suppose $n_i=1$ and for a particular equilibrium condition $x_e, u_e$ of \eqref{eq:general system}, the local linearization $\dot z = A(x_e, u_e)z+ B(x_e)v$ admits a stabilizing feedback gain $K$ such that the closed-loop system matrix $\dot z = (A(x_e, u_e) + B(x_e)K)z$ is positive.	Then in a neighborhood of $(x_e, u_e)$ there exists a sum-separable contraction metric satisfying the conditions of Theorem \ref{thm:main result}.
\end{theorem}
\begin{proof}
The linear closed-loop system has a diagonal quadratic Lyapunov function $z^TPz$ taking the metric with $M=P$ and differential feedback, $\delta_u=K\delta_x$, \eqref{eq:MI:SSCCM} therefore holds at $x_e, u_e$. Since it is a strict inequality and $A, B$ are smooth functions of $x,u$, it holds in a  neighborhood of $(x_e, u_e)$.
\end{proof}
The natural nonlinear analogue of a positive system is a {\em monotone} system \cite{smith_monotone_1995}, which preserves element-wise ordering between pairs of solutions, though for monotone systems the question of the existence of a separable Lyapunov function is more subtle \cite{Dirr2015}. In \cite{Manchester2017existence} global existence of separable contraction metrics was shown for certain classes of monotone contracting nonlinear systems. In addition, the utility of naturally-separable $l^1$-type metrics have been used by several authors in the analysis of monotone system \cite{como2015throughput, coogan2016separability}. Beyond these results, to the authors' knowledge the question of how restrictive it is to require $M$ to be separable remains open.

\begin{theorem}\label{thm:necessary}
Suppose $n_i=1$ for $i\in \mathbb{N}_{[1,N]}$ and suppose there exists a feedback controller $u(t)=k(x(t),x^\star(t),u^\star(t))$ solving Problem \ref{problem formulation} such that the closed-loop system
$  	\dot x = f(x,k(x,x^\star,u^\star))
$
 is:
  \begin{enumerate}
  \item contracting with respect to a constant metric $M>0$, i.e.  \begin{equation}\label{eq:monotone cl contracting}
  M(A+BK)+(A+BK)^TM\prec -2\lambda M
  \end{equation}
  for all $x, x^\star, u^\star$, where $K = \frac{\partial k}{\partial x}$,
  	\item monotone: $(A+BK)_{ij}\ge 0$ for $i\ne j$,
  	\item linearly coupled: $(A+BK)_{ij}$ is independent of $x$ for $i\ne j$.
  \end{enumerate}
Then there exists a sum-separable contraction metric satisfying the conditions of Theorem \ref{thm:main result}.  
\end{theorem}

\begin{proof}
Since the closed-loop system is contracting, monotone, and has linear coupling, by \cite[Theorem 6]{Manchester2017existence} it has a separable contraction metric.

Now, by assumption \eqref{eq:monotone cl contracting} holds for all $x, x^\star, u^\star$ for the closed-loop system, i.e. with
$A = A(x,k(x,x^\star, u^\star))$. In particular, it holds when $x=x^\star$, for which $u=k(x,x^\star,u^\star)=u^\star$. So for any $x^\star, u^\star$, 

This implies that \eqref{eq:CCM MK} holds for all $x,u$, hence $M$ is a separable control contraction metric for \eqref{eq:general system}.
\end{proof}
}

\section{Scalable Design of Distributed Controllers}\label{sec:distributed_design}

\mypara While the above developments give convex conditions for the design of distributed controllers, for large scale systems they may still be impractical. The problem is that one must find $W$ and $Y$ that satisfy \eqref{eq:MI:SSCCM}, which is a matrix inequality of the same dimension of the total number of states in the full network. Despite its sparsity, this can still be very challenging to solve.

In this section we show that when the combined communication/physical interconnection graph is \textit{chordal}, the problem of solving \eqref{eq:MI:SSCCM} is dramatically simplified. Many engineering systems naturally have chordal graph structures, and this has motivated research in efficient methods for semidefinite and sum-of-squares programming \cite{VandenbergheAndersen2015, PakazadHanssonAndersenEtAl2015, Waki2006}.


\tb{
\begin{theorem}\label{prop:clique tree decomposition}
	Let $\mathscr{G}:=\mathscr{G}_p\cup\mathscr{G}_c$ and suppose $\mathscr{G}^u$ is chordal. Let $l\in\mathbb{N}$ be the number of nodes of the clique tree $\mathscr{T}(\mathscr{G}^u)$. Then, the pointwise LMI~\eqref{eq:MI:SSCCM} can be decomposed into $l$ pointwise  LMIs of smaller dimension, each corresponding to a clique. Furthermore, each pointwise LMI depends only on the  $x_i, \breve{x}_i$ and $\vec{x}_i$ for each node $i$ contained in the corresponding clique.
\end{theorem}}

\begin{proof}
	Using the Algorithm~3.1 from \cite{VandenbergheAndersen2015}, it is possible to decompose the graph $\mathscr{G}$ into the clique tree $\mathscr{T}(\mathscr{G})$. Let the integer $l>0$ be the number of cliques in $\mathscr{T}(\mathscr{G})$. Our proof follows similar arguments to Section~II of \cite{PakazadHanssonAndersenEtAl2015}.
	
\mypara 	Let the sets $\mathscr{C}_1,\ldots,\mathscr{C}_l$ be the nodes of $\mathscr{T}(\mathscr{G})$, and $\mathbin{\mathtt{card}}_k$ be cardinality (number of elements) of the set $\mathscr{C}_k$, $k\in\mathbb{N}_{[1,l]}$. For each index $k\in\mathbb{N}_{[1,l]}$, define the matrix $E_k\in\mathbb{R}^{\mathbin{\mathtt{card}}_k\times n}$ obtained from the $n\times n$ identity matrix with blocks of rows indexed by $\mathbb{N}_{[1,N]}\setminus \mathscr{C}_k$ removed.
	
\mypara    Denote the left-hand side of the LMI~\eqref{eq:MI:SSCCM} by $T$. The existence of $l$ cliques implies that there exist matrices $F_k:\mathbb{R}^{\mathbin{\mathtt{card}}_k}\to\mathbb{R}^{\mathbin{\mathtt{card}}_k\times \mathbin{\mathtt{card}}_k}$, where $k\in\mathbb{N}_{[1,l]}$, satisfying   \begin{equation}\label{eqn:T_decomp}
   	 T=\sum_{k=1}^l E_k^T F_k E_k\;.
  	\end{equation}
Then if $
  		F_k\prec0,\quad\forall k\in\mathbb{N}_{[1,l]},
	$ the matrix $T$ is negative definite. Thus, the LMI~\eqref{eq:MI:SSCCM} holds. 
	
	For each node $i\in\mathscr{V}$ contained in the clique $\mathscr{C}_k$. The corresponding matrix $F_k$ has arguments $x_i, \breve{x}_i$ and $\vec{x}_i$. In other words, $F_k$ depends on how strongly the nodes of the  system (defined by $\mathscr{G}_c$) and communication network (defined by $\mathscr{G}_p$) are connected among each other.
\end{proof}

\mypara \tbb{If a graph is not chordal, it is possible to make it chordal by adding ``fake'' edges to form new cliques in the graph. This is referred to as a \textit{chordal embedding}, \textit{chordal extension}, or a \textit{triangulation}. Algorithms for finding such triangulations are well-developed and widely-used for solving large sparse linear equations and semidefinite programs \cite{heggernes2006minimal, VandenbergheAndersen2015}. 

We note here that these ``fake'' edges are only used to define the $l$ cliques used in the decomposition \eqref{eqn:T_decomp}, in order to speed up the computational verification of \eqref{eq:MI:SSCCM}. The fake edges do not appear in the communication graph and do not have any impact on the resulting structure of the metric $M$ or differential controller $K$, and hence do not effect the theoretical results on stabilization or distributed communication structure.}

\section{Illustrative Examples}\label{sec:Illustration}

\tb{

\subsection{Distributed Control of a Vehicle Platoon}
We first illustrate the proposed method through the design of a distributed nonlinear platoon controller.
Platooning provides a means for improving road safety, throughput and vehicle efficiency.
The control objective is for groups of vehicles to cooperatively maintain a group reference velocity with small intervehicle spacing. 

Each vehicle is assumed to be equipped with a radar measuring intervehicle distance and a wireless communication device to communicate with surrounding vehicles. Limitations in range and delay in the communication device mean that all-to-all communication within a platoon is generally impossible, i.e. the platoon must operate with communication limited to nearby vehicles. Several authors have proposed distributed controllers achieving stability and string stability subject to communication constraints e.g. \cite{ploeg2014controller, dunbar2012distributed, zheng2017distributed} and references therein. In \cite{monteil2017design}, the use of a nonlinear protocol leads to significant improvements in string stability.

\begin{figure}
	\centering
	\includegraphics[ trim = {0.5cm 0.5cm 1cm 0cm}, clip, width = 0.49\linewidth]{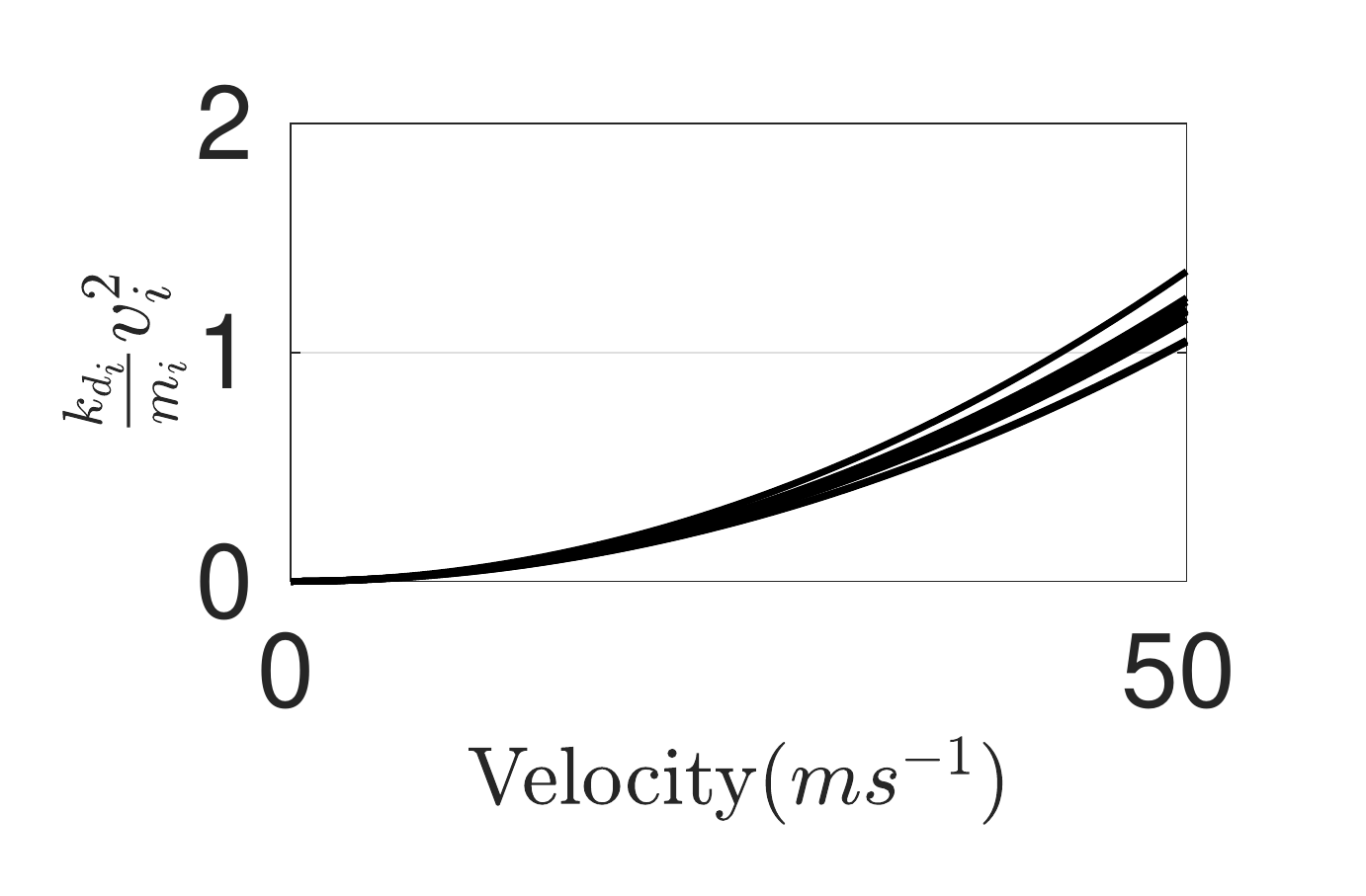}
	\includegraphics[ trim = {0.5cm 0.5cm 1cm 0cm}, clip, width = 0.49\linewidth]{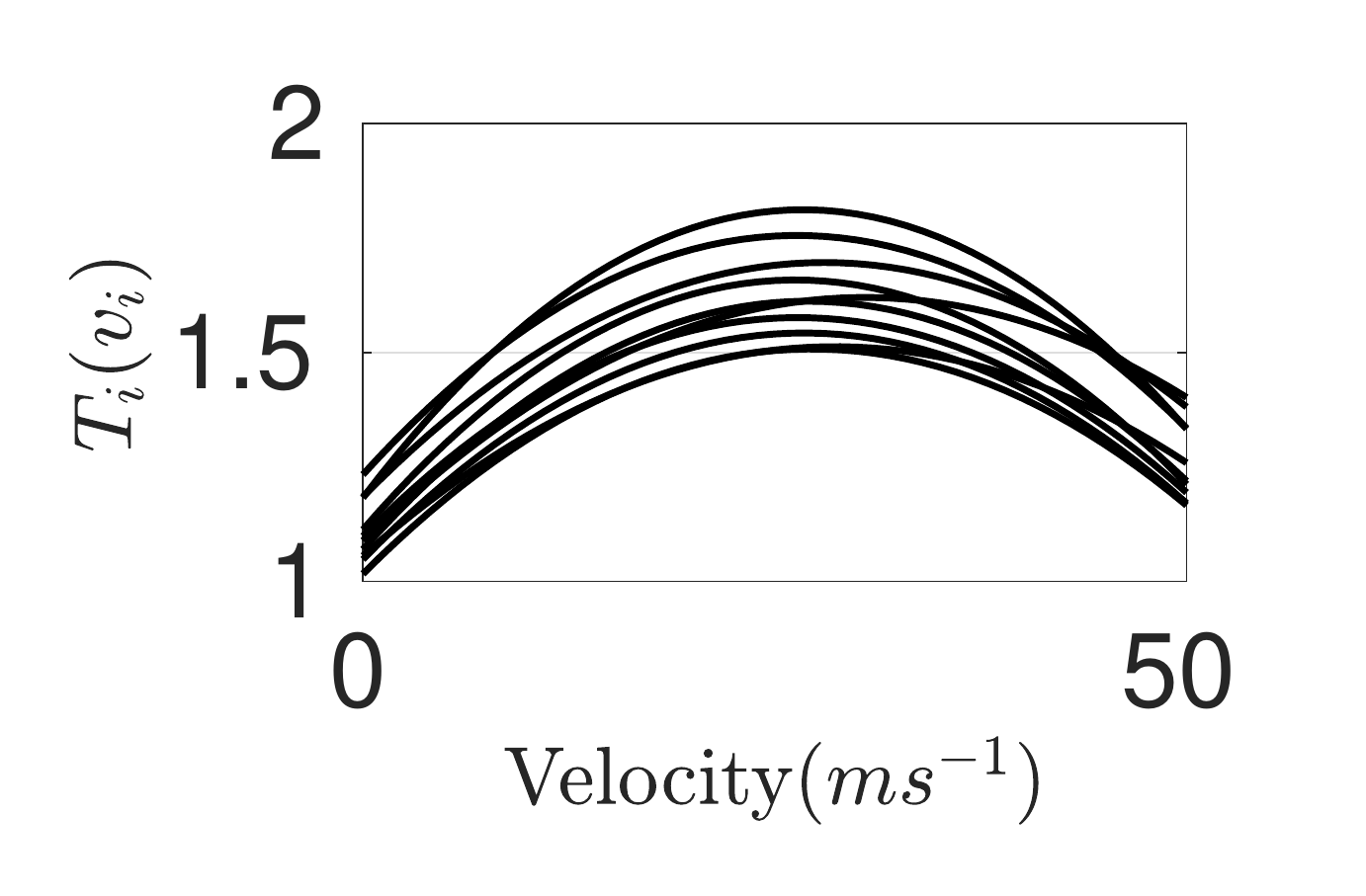}
	\caption{\label{fig:nonlinearity} Variation of drag force and drivetrain gain over velocity for each vehicle. }
\end{figure}

\begin{table}
	\centering
	\begin{tabular}{|c|c|c|c|c|c|}
	\hline	$m$ (tonne)& $k_{d}$ (kg/m)& $\alpha_i$ & $\beta_i$ & $\bar{\omega}$ (rad/s) & $\bar{T}$ (Nm) \\
	 \hline  1.8-2 & 1.3-1.6& 13-16 & 0.28-0.35 & 420  & 190 \\ \hline
	\end{tabular}
	\caption{\label{table: vehicle Parameters} Range of parameters used for vehicles in platoon.}
\end{table}


Adapting the model used in \cite[Sec 3.1]{astrom2010feedback}, we design decentralized controllers for platoons of heterogenous vehicles with dynamics 
\begin{equation} \label{eq:vehicle_dynamics}
	\dot{s}_i=v_i, \quad\quad
	\dot{v}_i=  \frac{1}{m_i} T_i(v_i)u_i - \frac{k_{d_i}}{2m_i} v_i^2 + \omega_i .
\end{equation}
where $s_i$, $v_i$, $u_i$ and $\omega_i$ are the $i$th vehicles position, velocity, control input and a disturbance. The term $T_i(v_i)$ represents the dynamics of the drive chain 
\begin{equation*} \label{eq:drivechain_dynamics}
T_i(v_i) = \alpha_i T_{m_i} \left( 1 - \beta_i \left(\frac{\alpha_i v_i}{\omega_{m_i}} - 1\right)^2\right).
\end{equation*}
The parameters used are randomly selected from the range shown in table \ref{table: vehicle Parameters}. Choosing a state vector of $x = (s_1, v_1, s_1-s_2, v_2,..., s_{N-1} - s_N, v_N)^T$ allows for the problem of platooning at a constant velocity with constant spacing to be formulated as tracking a trajectory $x(t) = (v^*t, v^*, d^*, v^*,..., d^*, v^*)$ where $v^*$ is the desired nominal platoon velocity and $d^*$ is the intervehicular spacing. The dynamics of the platoon are written concisely in the form \eqref{eq:stacked system}.

We consider a balanced communication graph with a horizon $h$. That is, each agent $i$ has access to the state of agents $ j \in \mathbb{N}_{[i-h, i+h]}$ with $ 1 \leq j \leq N$. 

One advantage of the convexity of CCM synthesis is the ease of adding additional constraints. In this paper, we constrain the nonlinear controller to match a prescribed linear $H^\infty$ controller near a particular operating point.  

\subsection*{Distributed Linear $H^\infty$ Control Design}
Choosing a nominal operating point of $v^* = 25ms^{-1}$, $u_i^* = \frac{ k_{d_i} {v^*}^2 }{2 T(v^*)}$, 
we define the linearized system 
\begin{gather}
\dot{x} = \tilde A x + \tilde Bu + H w\\
y = \tilde Cx + \tilde D u \nonumber
\end{gather}
where $\tilde A = \frac{\partial f}{\partial x}|_{v^*, u^*}$, $\tilde B = \frac{\partial f}{\partial u}|_{v^*, u^*}$, $B_{w} = (0,1,0,...0)^T$ and $\tilde C, \tilde D$ specify the performance output which are chosen to be:
\begin{align*}
&y_{v_1} = q_{v_1}(v_1 ), & y_{s_1} = q_{s_1}(s_1),& &y_{u_1} = q_{u_1}(u_1) \\
& y_{s_i} = q_s(s_{i-1} - s_i), &y_{u_i} = q_u(u_i) & &
\end{align*}
for $i = 2,...,N$ where $q_s$, $q_v$ and $q_u$ are weights used to tune the controller. The values used for the examples in this paper were $(q_{v_1},q_{s_1}, q_{u_1}, q_s, q_u) = (10^{-2},1,3\times10^5,10^3,5\times10^4)$.

We assume the existence of a block diagonal storage function $V(x) = x^T P x$ rendering the structured controller design problem convex. While the restriction to a block diagonal $P$ is generally conservative, we find the same resulting gain bound for the cases when $P$ is full and $P$ is block diagonal.

We solve a state-feedback $H^\infty$ control problem by searching for a storage function $P = Q^{-1}$ and feedback gain $K = ZQ \in \Xi$ that minimizes a performance bound $\sup_w \frac{||y||_{\mathcal{L}_2}}{ ||w||_{\mathcal{L}_2} } \leq \alpha$ via the following semidefinite program \cite{tanaka2011bounded}:
\begin{equation} \label{eq:Linear Control Design}
\begin{aligned}
		& \underset{Q, Z, \alpha}{\text{minimize}} & & \alpha \\
		& \text{subject to} &  & Q > 0 , \quad Z \in \Xi,\\
\end{aligned}
\end{equation}
$$
\begin{bmatrix}
\tilde AQ + \tilde BZ + (\tilde AQ+ \tilde BZ)^T & H & (\tilde CQ + \tilde D Z)^T\\
H^T & -\alpha I & 0\\
\tilde CQ + \tilde D Z & 0& -\alpha I
\end{bmatrix} < 0
$$
In general, there are many controllers that can satisfy the same gain bound in problem \eqref{eq:Linear Control Design}. As such, we can improve performance by first solving \eqref{eq:Linear Control Design} and then fixing $\alpha$ and maximize the smallest eigenvalue of $Q$.

\subsection*{Distributed CCM}

The set of matrices $W,Y$ satisfying LMIs \eqref{eq:MI:SSCCM} define a set of universally stabilizing control laws of the form \eqref{eq:contracting feedback law}. Note however, that the model exhibits non-physical behaviour for negative or large $v_i$, when the term $T_i(v_i)$ is zero or negative, hence LMI \eqref{eq:MI:SSCCM} cannot be satisfied over all $x$. It can, however,  be enforced over a convex set $v_i \in [0, 50 ms^{-1}]$ using Lagrange multipliers, c.f. Remark \ref{rem:convex} above.
We utilize a dummy variable $\nu$ to help solve the following feasibility problem 
\begin{gather*}
W \succ 0, \quad Y(x_{nom}) = KW(x_{nom}) \\
-\nu^T \left( AW + WA^T + BY + (BY)^T + 2 \lambda W\right) \nu \\ -\sum_i \tau_iv_i(v_i-50) \succ 0,
\end{gather*}
where $Y\in \Xi$ consists of degree 2 polynomials, $W$ is a block diagonal, flat metric and $\tau_i$ is a lagrangian multiplier consisting of degree 2 polynomials in $x$ and $\nu$.
Solving this problem with $N = 10$ and $\lambda = 0.02$ using Yalmip \cite{Loefberg2004,Loefberg2009} and Mosek on an intel i7 processor with 8GB of ram took 9 seconds for $h=0$ and 40 seconds for $h = 1$.

We compare the resulting controllers for two communication patterns in three different situations. The first situation looks at tracking a step change in reference velocity from $10ms^{-1}$ to $5ms^{-1}$ that occurs at time $t = 5$ seconds. We then study the platoon response to a temporary disturbance at time $t=10$  and a worst-case step disturbance at time $t = 20$ as described by \eqref{eq:platoon disturbance}. The platoon velocity response can be seen in figure \ref{fig:velocity disturbance response} and the platoon's position response can be seen in figure \ref{fig:position disturbance response}.
\begin{equation} \label{eq:platoon disturbance}
w_1(t) = \begin{cases}
20\sin(\frac{2\pi}{10}(t - 95)), & 95 \leq t \leq 100\\
10 & t\geq 180 \\
0 & \text{otherwise}
\end{cases}
\end{equation}

Figures \ref{fig:velocity disturbance response} and \ref{fig:position disturbance response} show the well known, desirable effects that increasing communication has on the rate of synchronization and propogation of disturbances down the vehicle chain. Figure \ref{fig:position disturbance response} also shows an overall reduction in the magnitude of the disturbance response. Note that the nonlinear system is operating far from the linearization point of 25$ms^{-1}$. The use of separable control contraction metrics, allows for controllers with different communication patterns to be easily developed with guaranteed stability across an operating range.

\begin{figure}
	\begin{minipage}{1\linewidth}
		\centering
		\includegraphics[trim = {0.5cm 0cm 1cm 0.5cm}, clip, width=\linewidth]{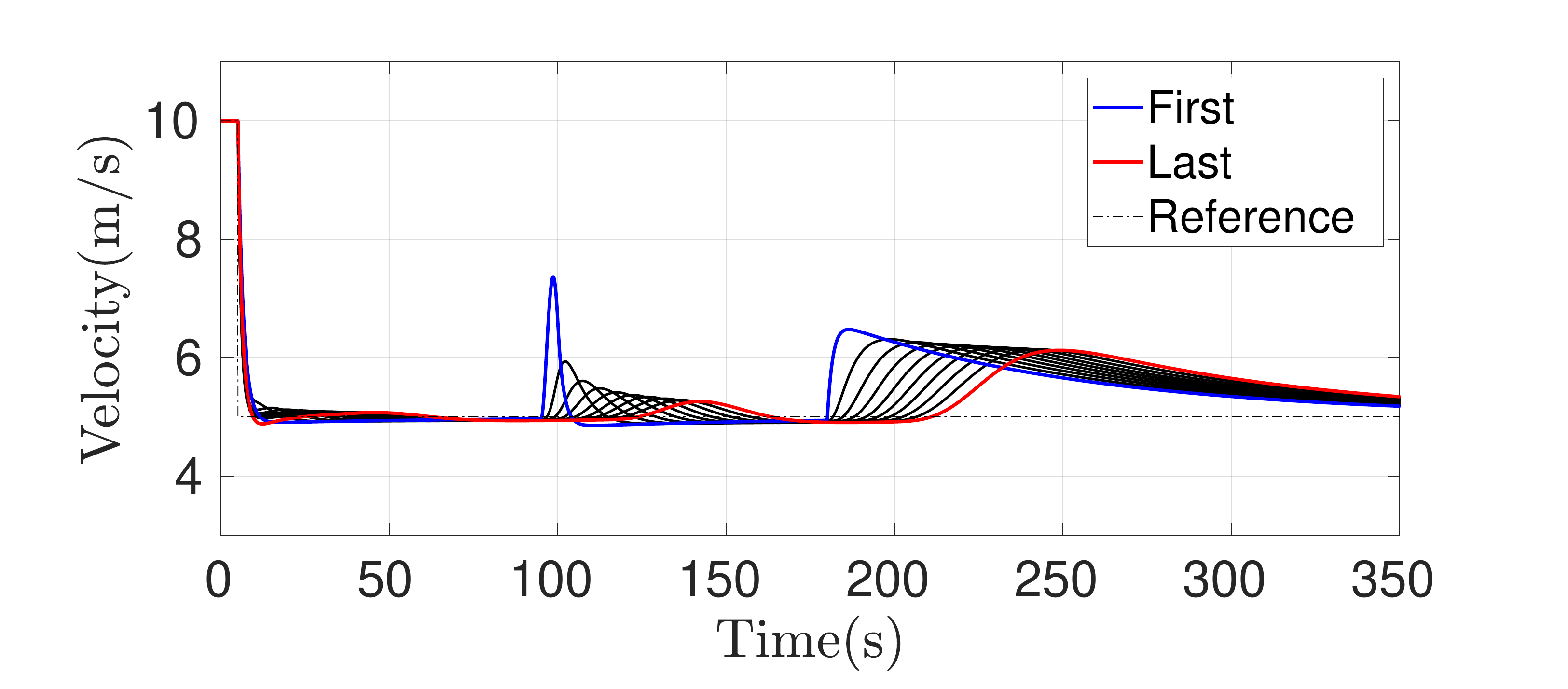}
		\subcaption{$h = 0$}
	\end{minipage}
	\begin{minipage}{1\linewidth}
		\centering
		\includegraphics[trim = {0.5cm 0cm 1cm 0.5cm}, clip, width=\linewidth]{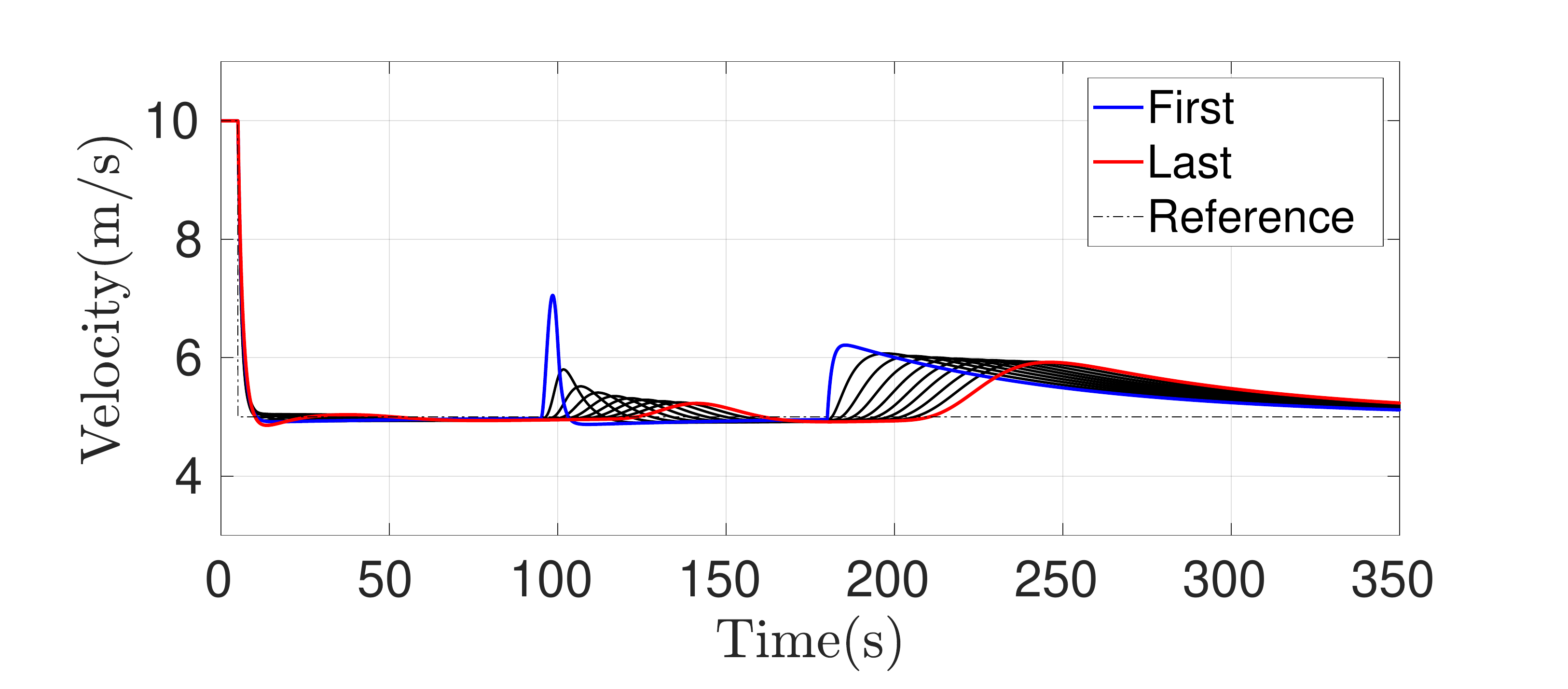}
		\subcaption{$h = 1$}
	\end{minipage}
	\caption{\label{fig:velocity disturbance response}Velocity response of 10 car platoon to step reference change and disturbance \eqref{eq:platoon disturbance}. The first vehicle is in blue and the last is in red.}
\end{figure}
\begin{figure}[htpb!]
	\centering
	\begin{minipage}{1\linewidth}
		\includegraphics[trim = {0.5cm 0cm 1cm 0.5cm}, clip, width=\linewidth]{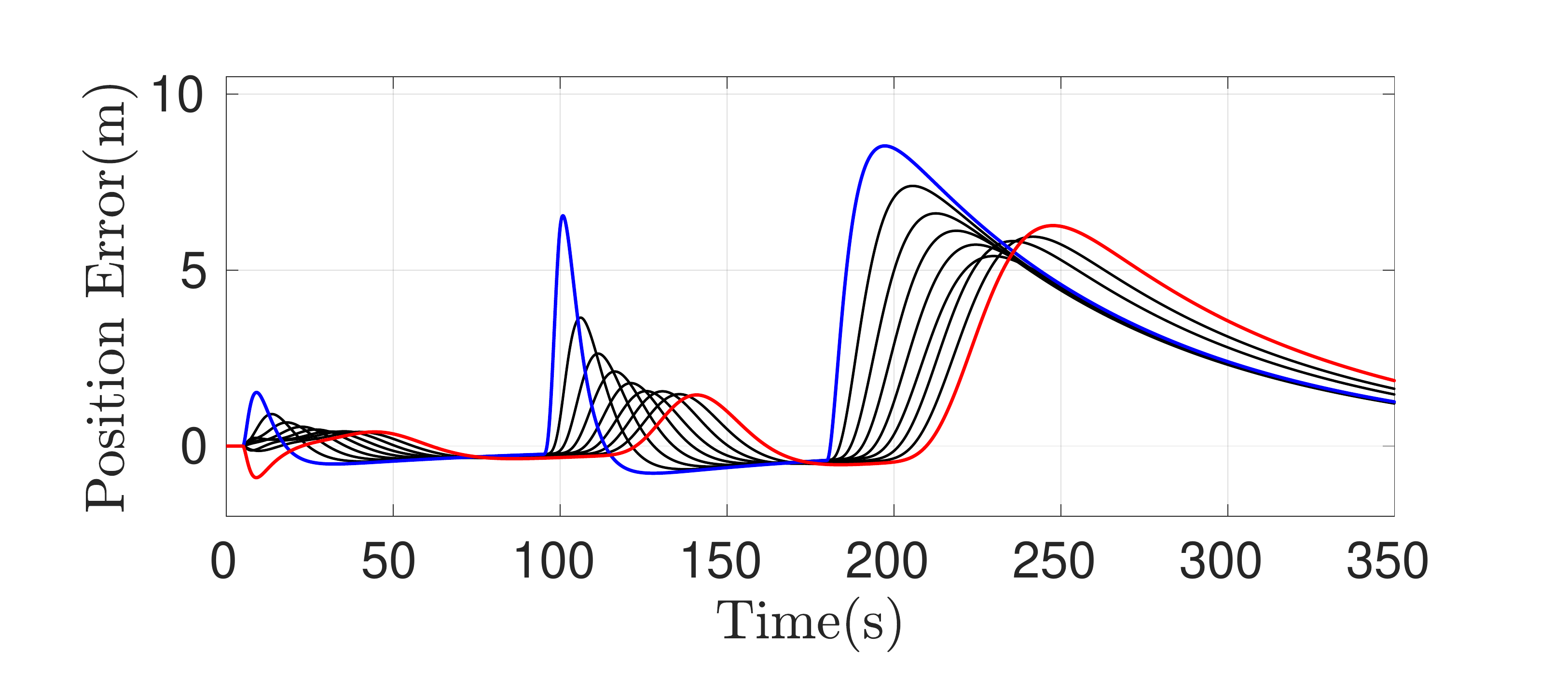}
		\subcaption{$h = 0$}
	\end{minipage}
	\begin{minipage}{1\linewidth}
		\includegraphics[trim = {0.5cm 0cm 1cm 0.5cm}, clip, width=\linewidth]{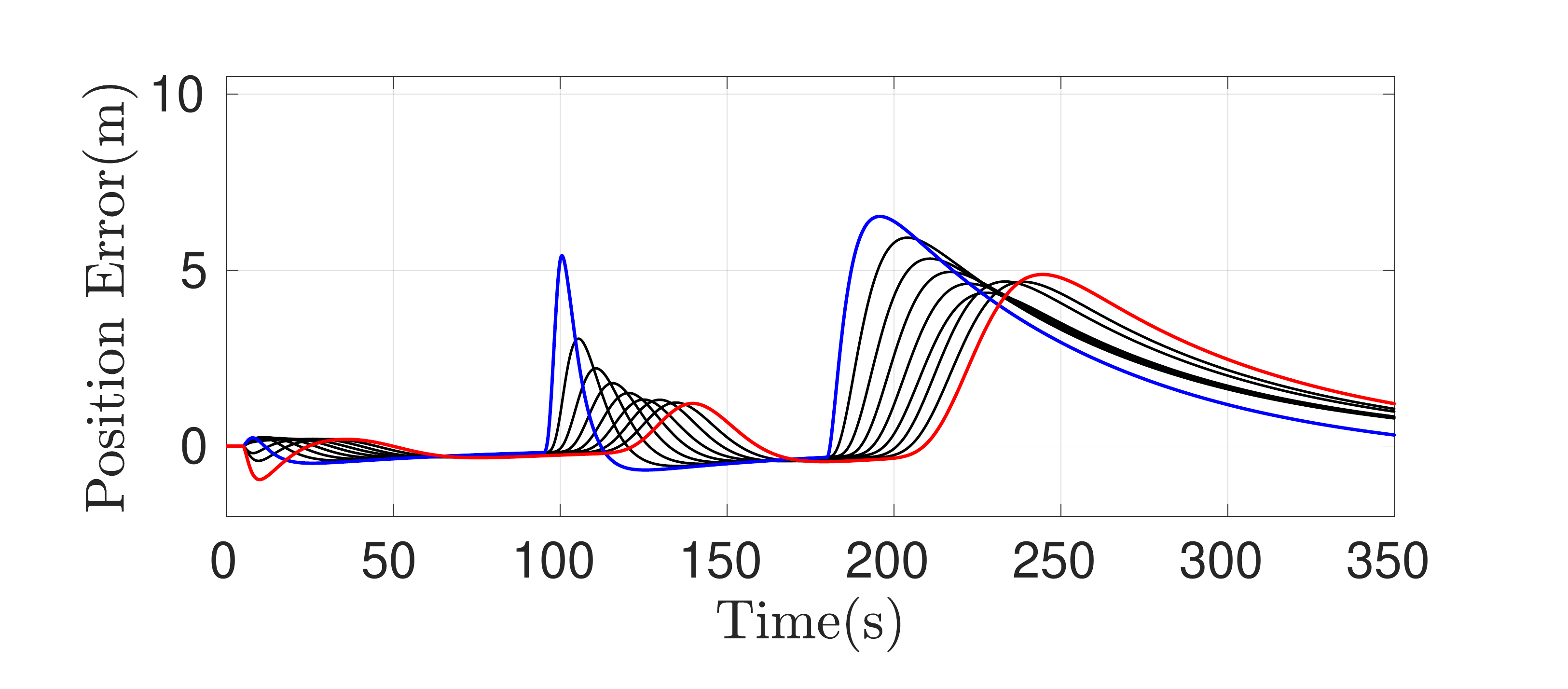}
		\subcaption{$h = 1$}
	\end{minipage}
	
	\caption{\label{fig:position disturbance response}Response of 10 car platoon to step reference change and disturbance \eqref{eq:platoon disturbance}. The first vehicle is in blue and the last is in red. The position error is taken as $s_{i-1} - s_i - d^*_i$ for  $i = 2,...,N$}
\end{figure}

}

\subsection{Scalability and Flexibility: Large-Scale System with Uncontrollable Linearization}\label{sec:scalability example}

In this subsection, we consider a more academic example to illustrate the flexibility and scalability of the CCM approach. Consider a system of $N$ agents with local dynamics
\begin{equation}\label{eq:example}
	\scalebox{0.95}{$\begin{array}{rcl}
		\dot{x}_i&=&-x_i-x_i^3+y_i^2 + 0.01\left(x_{i-1}^3 - 2x_i^3 +x_{i+1}^3\right)\\
		\dot{y}_i&=&u_i\;,
	\end{array}$}
\end{equation}
for $i\in\mathbb{N}_{[1,N]}$ and for convenience define the boundary states $x_0=x_1$ and $x_N=x_{N+1}$. For each index $i\in\mathbb{N}_{[1,N]}$, define the vectors $q_i=(x_i,y_i)$, $\breve{q}_i=(x_{i-1},x_{i+1})$ and let $q=(q_1,\ldots,q_N)$, and
\begin{align*}
	f_i(q_i,\breve{q}_i)=&\begin{bmatrix}
		-x_i-x_i^3+y_i^2 + 0.01\left(x_{i-1}^3 - 2x_i^3 +x_{i+1}^3\right)\\
		0
	\end{bmatrix}\\
	B_i=&\begin{bmatrix}
	0, 1
	\end{bmatrix}^T.
\end{align*}
Note that system \eqref{eq:example} is not controllable when linearized about the origin, since the $x$ and $y$ dynamics are decoupled, and furthermore is not feedback linearizable in the sense of \cite{Isidori:1995}, because the vector fields
\begin{align*}
	B=&\mathbin{\diag}(B_1,\ldots,B_N),\\
	\dfrac{\partial f}{\partial q}B-\dfrac{\partial B}{\partial q}f=&\mathbin{\diag}\left(\begin{bmatrix}
	2y_1\\0
	\end{bmatrix},\ldots,\begin{bmatrix}
	2y_N\\0
	\end{bmatrix}\right)
\end{align*}
are not linearly independent at the origin.  Furthermore, due to the quadratic term on $y$, the only possible action by the controller on the $x$-subsystem is to move the $x$-component of solution to \eqref{eq:example} towards the positive semi-axis. In other words, the controller cannot reduce the value of the $x$-component.






\mypara To show the advantages of the method proposed in this paper, a benchmark composed of  three scenarios, according to the constraints imposed on the matrix $Y$, has been created.  Namely, the unconstrained case, in which $\mathscr{G}_c$ is a complete graph, the ``neighbor'' case, in which $\mathscr{G}_c=\mathscr{G}_p$, and the fully decentralized case, in which $\mathscr{G}_c$ has no edges $(i,j)$ with $i\ne j$. 

In each case we searched for a constant dual metric $W$ and a matrix function $Y$ with second-order polynomial terms in the variables as described by $\Xi$. The numerical results were obtained using Yalmip \cite{Loefberg2004,Loefberg2009} and Mosek running on an Intel Core i7 with 32GB RAM.

\mypara For the unconstrained case, the graph $\mathscr{G}_c$ describing the communication network is fully connected and the matrix $Y$ was full, with each element able to depend on all state variables. 
For this case, the set of matrix inequalities~\eqref{eq:MI:SSCCM} could not be solved due to memory constraints when $N>8$, i.e. state dimension $n>16$.

For the two latter cases, it was possible to solve \eqref{eq:MI:SSCCM} for up to $N=512$ systems, i.e. a full state dimension of $n=1024$, using the chordal decomposition of Section \ref{sec:distributed_design}. Since the string topology is chordal, and the LMI~\eqref{eq:MI:SSCCM} can be decomposed into $N-1$ cliques each with two nodes.

\mypara Figure \ref{fig:3dsim} shows simulations of the network~\eqref{eq:example} with $N=4$. All controller structures achieve exponential convergence, whereas the open loop simulation (performed with $u\equiv0$) does not converge to the origin.
\begin{figure}
	\centering
	\includegraphics[width=\linewidth]{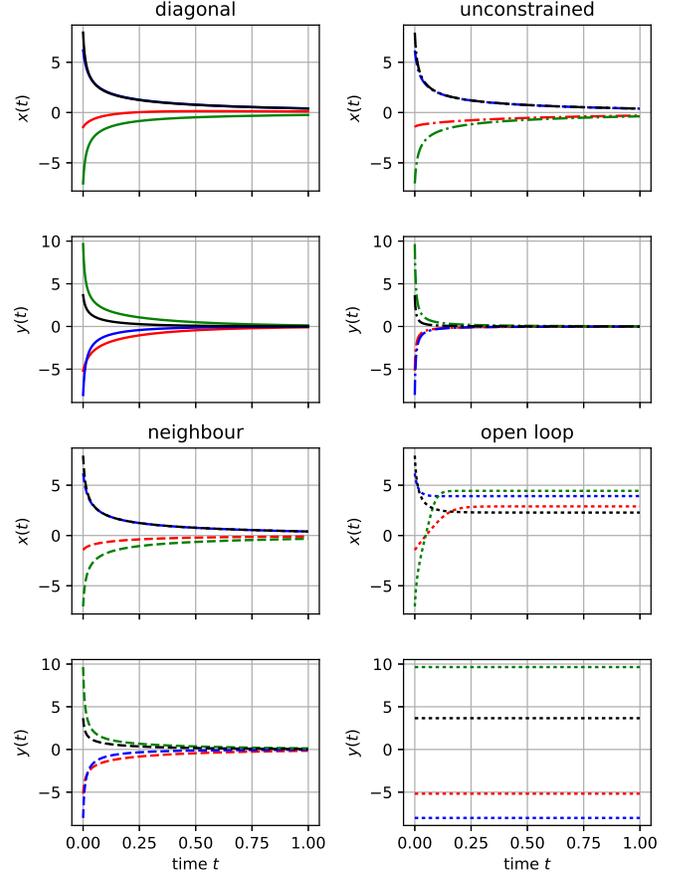}

	\caption{Simulation of network~\eqref{eq:example} with $N=4$ and the target trajectory being the origin and under the controller obtained according to different constrains for the computation of $Y$: diagonal, neighbour and unconstrained.}
	\label{fig:3dsim}
\end{figure}
%

\mypara Figure \ref{fig:time graph} shows a plot of the time taken to solve \eqref{eq:MI:SSCCM} for the three cases considered in this topology: unconstrained, ``neighbor'' and fully decentralized.  According to this graph, for $N=1,2$, the time taken for each of the three cases is comparable. However, as the number of systems increases, the unconstrained quickly becomes infeasible, whereas the neighbor and  decentralized cases, the computation time is approximately linear in the number of nodes. 


\begin{figure}
	\centering
	\includegraphics[width=\columnwidth]{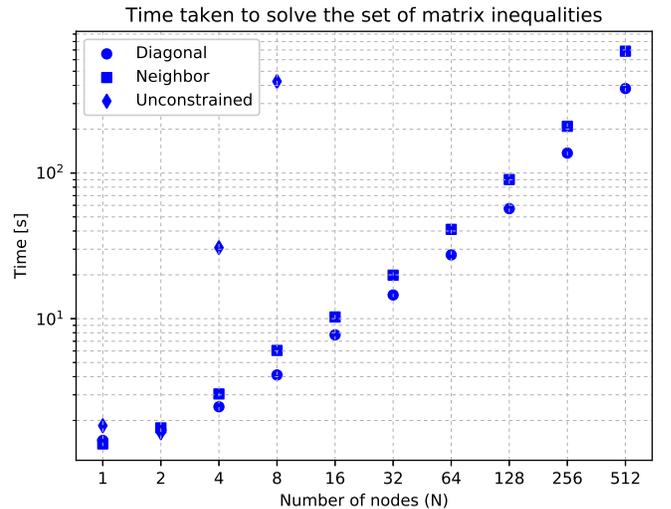}
	\caption{Computation time required to solve for a CCM for the three communication structures considered in Section \ref{sec:scalability example}.}
	\label{fig:time graph}
\end{figure}

\section{Conclusions}\label{sec:Conclusion}

\mypara In this paper we have developed a method for control design using separable control contraction metrics, building upon \cite{manchester_control_2017}. The main advantage in using a separable CCM is that it allows a convex (semidefinite programming) search for nonlinear feedback controllers with specified communication structure in the controller. Furthermore, we have shown that the search for a CCM can be made scalable for certain interaction structures defined by chordal graphs.



%

\bibliographystyle{ieeetr}
\bibliography{Library,Distributed}

\begin{thebibliography}{10}

\bibitem{hill_smart_2012}
D.~J. Hill, T.~Liu, and G.~Verbic, ``Smart grids as distributed learning
  control,'' in {\em 2012 {{IEEE Power}} and {{Energy Society General
  Meeting}}}, pp.~1--8, July 2012.

\bibitem{Pajic2011}
M.~Pajic, S.~Sundaram, G.~J. Pappas, and R.~Mangharam, ``The wireless control
  network: A new approach for control over networks,'' {\em {IEEE} {T}rans.
  {A}utom. {C}ontrol}, vol.~56, pp.~2305--2318, oct 2011.

\bibitem{wang_implementing_2016}
S.~Wang, J.~Wan, D.~Li, and C.~Zhang, ``Implementing {{Smart Factory}} of
  {{Industrie}} 4.0: {{An Outlook}},'' {\em International Journal of
  Distributed Sensor Networks}, vol.~12, p.~3159805, Jan. 2016.

\bibitem{CanudasdeWitMorbidiLeonOjedaEtAl2015}
C.~{Canudas de Wit}, F.~Morbidi, L.~{Leon Ojeda}, A.~Y. Kibangou, I.~Bellicot,
  and P.~Bellemain, ``Grenoble {T}raffic {L}ab: {A}n {E}xperimental {P}latform
  for {A}dvanced {T}raffic {M}onitoring and {F}orecasting,'' {\em {IEEE}
  {C}ontrol {S}ystems {M}agazine}, vol.~35, pp.~23--39, jun 2015.

\bibitem{anderson1990optimal}
B.~D.~O. Anderson and J.~B. Moore, {\em Optimal Control: Linear Quadratic
  Methods}.
\newblock {Prentice-Hall}, 1990.
\newblock 02941.

\bibitem{dullerud2000course}
G.~Dullerud and F.~Paganini, ``A {{Course}} in {{Robust Control Theory}}: {{A
  Convex Approach}},'' 2000.

\bibitem{sandell_survey_1978}
N.~Sandell, P.~Varaiya, M.~Athans, and M.~Safonov, ``Survey of decentralized
  control methods for large scale systems,'' {\em IEEE Transactions on
  Automatic Control}, vol.~23, pp.~108--128, Apr. 1978.

\bibitem{siljak1978large}
D.~D. {\v{S}}iljak, {\em Large-scale dynamic systems: stability and structure}.
\newblock North Holland, 1978.

\bibitem{BlondelTsitsiklis1997}
V.~Blondel and J.~N. Tsitsiklis, ``{NP}-{H}ardness of {S}ome {L}inear {C}ontrol
  {D}esign {P}roblems,'' {\em {SIAM} {J}. {C}ontrol {O}ptim.}, vol.~35,
  pp.~2118--2127, nov 1997.

\bibitem{zecevic_control_2010}
A.~Zecevic and D.~D. Siljak, {\em Control of {{Complex Systems}}: {{Structural
  Constraints}} and {{Uncertainty}}}.
\newblock {Springer}, Jan. 2010.

\bibitem{Tanaka2011}
T.~Tanaka and C.~Langbort, ``The bounded real lemma for internally positive
  systems and h-infinity structured static state feedback,'' {\em {IEEE}
  {T}rans. {A}utom. {C}ontrol}, vol.~56, pp.~2218--2223, Sep 2011.

\bibitem{Rantzer2015}
A.~Rantzer, ``Scalable control of positive systems,'' {\em European Journal of
  Control}, vol.~24, pp.~72--80, jul 2015.

\bibitem{berman_nonnegative_1994}
A.~Berman and R.~J. Plemmons, {\em Nonnegative {{Matrices}} in the
  {{Mathematical Sciences}}}.
\newblock {SIAM}, 1994.

\bibitem{umenberger_scalable_2016}
J.~Umenberger and I.~R. Manchester, ``Scalable {{Identification}} of {{Positive
  Linear Systems}},'' in {\em Proceedings of the 55th {{IEEE Conference}} on
  {{Decision}} and {{Control}}}, (Las Vegas, NV), Dec. 2016.

\bibitem{Slotine1991}
J.~Slotine and W.~Li, {\em Applied Nonlinear Control}.
\newblock Prentice Hall, 1991.

\bibitem{Krstic1995}
M.~Krsti\'{c}, I.~Kanellakopoulos, and P.~V. Kokotovi\'{c}, {\em Nonlinear and
  Adaptive Control Design}.
\newblock Wiley, 1995.

\bibitem{Khalil:2001}
H.~K. Khalil, {\em Nonlinear Systems}.
\newblock Prentice Hall, 3rd~ed., 2001.

\bibitem{Rantzer:2001}
A.~Rantzer, ``A dual to {L}yapunov's stability theorem,'' {\em Syst. \& Contr.
  Lett.}, vol.~42, pp.~161--168, 2001.

\bibitem{Lohmiller1998}
W.~Lohmiller and J.-J. Slotine, ``On contraction analysis for nonlinear
  systems,'' {\em Automatica}, vol.~34, no.~6, pp.~683--696, 1998.

\bibitem{Angeli2002}
D.~Angeli, ``A {L}yapunov approach to incremental stability properties,'' {\em
  {IEEE} {T}rans. {A}utom. {C}ontrol}, vol.~47, pp.~410--421, Mar. 2002.

\bibitem{Wang2004}
W.~Wang and J.-J.~E. Slotine, ``On partial contraction analysis for coupled
  nonlinear oscillators,'' {\em Biological Cybernetics}, vol.~92, pp.~38--53,
  Dec 2004.

\bibitem{PhamSlotine2007}
Q.~Pham and J.~Slotine, ``Stable concurrent synchronization in dynamic system
  networks,'' {\em Neural Networks}, vol.~20, pp.~62--77, jan 2007.

\bibitem{russo2010global}
G.~Russo, M.~Di~Bernardo, and E.~D. Sontag, ``Global entrainment of
  transcriptional systems to periodic inputs,'' {\em PLoS computational
  biology}, vol.~6, no.~4, p.~e1000739, 2010.

\bibitem{Aminzare2014a}
Z.~Aminzare and E.~D. Sontag, ``Synchronization of diffusively-connected
  nonlinear systems: Results based on contractions with respect to general
  norms,'' {\em {IEEE} Transactions on Network Science and Engineering},
  vol.~1, no.~2, pp.~91--106, 2014.

\bibitem{Russo2013}
G.~Russo, M.~di~Bernardo, and E.~D. Sontag, ``A contraction approach to the
  hierarchical analysis and design of networked systems,'' {\em {IEEE} {T}rans.
  {A}utom. {C}ontrol}, vol.~58, pp.~1328--1331, may 2013.

\bibitem{como2015throughput}
G.~Como, E.~Lovisari, and K.~Savla, ``Throughput optimality and overload
  behavior of dynamical flow networks under monotone distributed routing,''
  {\em IEEE Transactions on Control of Network Systems}, vol.~2, no.~1,
  pp.~57--67, 2015.

\bibitem{coogan2016separability}
S.~Coogan, ``Separability of {L}yapunov functions for contractive monotone
  systems,'' in {\em Decision and Control (CDC), 2016 IEEE 55th Conference on},
  pp.~2184--2189, IEEE, 2016.

\bibitem{Manchester2017existence}
I.~R. Manchester and J.~Slotine, ``On existence of separable contraction
  metrics for monotone nonlinear systems,'' in {\em Proc. of the 18th {IFAC}
  World Congress}, pp.~8226 -- 8231, July 2017.

\bibitem{arcak2011certifying}
M.~Arcak, ``Certifying spatially uniform behavior in reaction--diffusion {PDE}
  and compartmental {ODE} systems,'' {\em Automatica}, vol.~47, no.~6,
  pp.~1219--1229, 2011.

\bibitem{Manchester2014a}
I.~R. Manchester and J.-J.~E. Slotine, ``Control contraction metrics and
  universal stabilizability,'' in {\em Proceedings of the 19th {IFAC} World
  Congress}, no.~1, (Cape Town, South Africa), pp.~8223--8228, Aug 2014.

\bibitem{manchester_control_2017}
I.~R. Manchester and J.~J.~E. Slotine, ``Control {{Contraction Metrics}}:
  {{Convex}} and {{Intrinsic Criteria}} for {{Nonlinear Feedback Design}},''
  {\em IEEE Transactions on Automatic Control}, vol.~62, pp.~3046--3053, June
  2017.

\bibitem{SteinShiromotoManchester2016}
H.~{Stein Shiromoto} and I.~R. Manchester, ``Decentralized nonlinear feedback
  design with separable control contraction metrics,'' in {\em Proceedings of
  the 55th {C}onference on {D}ecision and {C}ontrol ({CDC})}, (Las Vegas, NV,
  USA), pp.~5551--5556, Dec 2016.

\bibitem{PakazadHanssonAndersenEtAl2015}
S.~K. Pakazad, A.~Hansson, M.~S. Andersen, and A.~Rantzer, ``Distributed
  semidefinite programming with application to large-scale system analysis,''
  {\em IEEE Transactions on Automatic Control}, vol.~63, no.~4, pp.~1045--1058,
  2018.

\bibitem{VandenbergheAndersen2015}
L.~Vandenberghe and M.~S. Andersen, ``Chordal {G}raphs and {S}emidefinite
  {O}ptimization,'' {\em Foundations and {T}rends in {O}ptimization}, vol.~1,
  no.~4, pp.~241--433, 2015.

\bibitem{Diestel2005}
R.~Diestel, {\em {G}raph {T}heory}, vol.~173 of {\em Graduate Texts in
  Mathematics}.
\newblock Springer, 2005.

\bibitem{Teschl2012}
G.~Teschl, {\em Ordinary differential equations and dynamical systems},
  vol.~1XX.
\newblock American Mathematical Society, 2012.

\bibitem{Boothby1986}
W.~M. Boothby, {\em An Introduction to Differentiable Manifolds and
  {R}iemannian Geometry}.
\newblock Academic Press, 1986.

\bibitem{Lewis1949}
D.~C. Lewis, ``Metric properties of differential equations,'' {\em American
  Journal of Mathematics}, vol.~71, pp.~294--312, apr 1949.

\bibitem{Forni2014}
F.~Forni and R.~Sepulchre, ``A {D}ifferential {L}yapunov {F}ramework for
  {C}ontraction {A}nalysis,'' {\em {IEEE} Transactions on Automatic Control},
  vol.~59, pp.~614--628, mar 2014.

\bibitem{Parrilo2003}
P.~A. Parrilo, ``Semidefinite programming relaxations for semialgebraic
  problems,'' {\em Mathematical Programming}, vol.~96, pp.~293--320, may 2003.

\bibitem{boyd1994linear}
S.~Boyd, L.~{el Ghaoui}, E.~Feron, and V.~Balakrishnan, {\em Linear Matrix
  Inequalities in System and Control Theory}.
\newblock {Society for Industrial and Applied Mathematics (SIAM)}, 1994.

\bibitem{Dirr2015}
G.~Dirr, H.~Ito, A.~Rantzer, and B.~S. R\"{u}ffer, ``Separable {L}yapunov
  functions for monotone systems: Constructions and limitations,'' {\em
  Discrete and Continuous Dynamical Systems Series B ({DCDS-B})}, vol.~20,
  pp.~2497--2526, aug 2015.

\bibitem{wang_distributed_2017}
R.~Wang, I.~R. Manchester, and J.~Bao, ``Distributed {{Economic MPC With
  Separable Control Contraction Metrics}},'' {\em IEEE Control Systems
  Letters}, vol.~1, pp.~104--109, July 2017.

\bibitem{smith_monotone_1995}
H.~L. Smith, {\em Monotone {{Dynamical Systems}}}.
\newblock No.~41 in Mathematical Surveys and Monographs, Providence, RI:
  {American Mathematical Society}, 1995.

\bibitem{Waki2006}
H.~Waki, S.~Kim, M.~Kojima, and M.~Muramatsu, ``Sums of squares and
  semidefinite program relaxations for polynomial optimization problems with
  structured sparsity,'' {\em SIAM J. Optim.}, vol.~17, pp.~218--242, Jan 2006.

\bibitem{heggernes2006minimal}
P.~Heggernes, ``Minimal triangulations of graphs: A survey,'' {\em Discrete
  Mathematics}, vol.~306, no.~3, pp.~297--317, 2006.

\bibitem{ploeg2014controller}
J.~Ploeg, D.~P. Shukla, N.~van~de Wouw, and H.~Nijmeijer, ``Controller
  synthesis for string stability of vehicle platoons,'' {\em IEEE Trans.
  Intelligent Transportation Systems}, vol.~15, no.~2, pp.~854--865, 2014.

\bibitem{dunbar2012distributed}
W.~B. Dunbar and D.~S. Caveney, ``Distributed receding horizon control of
  vehicle platoons: Stability and string stability,'' {\em IEEE Transactions on
  Automatic Control}, vol.~57, no.~3, pp.~620--633, 2012.

\bibitem{zheng2017distributed}
Y.~Zheng, S.~E. Li, K.~Li, F.~Borrelli, and J.~K. Hedrick, ``Distributed model
  predictive control for heterogeneous vehicle platoons under unidirectional
  topologies,'' {\em IEEE Transactions on Control Systems Technology}, vol.~25,
  no.~3, pp.~899--910, 2017.

\bibitem{monteil2017design}
J.~Monteil and G.~Russo, ``On the design of nonlinear distributed control
  protocols for platooning systems,'' {\em IEEE control systems letters},
  vol.~1, no.~1, pp.~140--145, 2017.

\bibitem{astrom2010feedback}
K.~J. Astr{\"o}m and R.~M. Murray, {\em Feedback systems: an introduction for
  scientists and engineers}.
\newblock Princeton university press, 2010.

\bibitem{tanaka2011bounded}
T.~Tanaka and C.~Langbort, ``The bounded real lemma for internally positive
  systems and h-infinity structured static state feedback,'' {\em IEEE
  transactions on automatic control}, vol.~56, no.~9, pp.~2218--2223, 2011.

\bibitem{Loefberg2004}
J.~L\"{o}fberg, ``{YALMIP} : a toolbox for modeling and optimization in
  {MATLAB},'' in {\em Proceedings of the {IEEE} International Symposium on
  Computer Aided Control Systems Design}, (Taipei, Rep. of China),
  pp.~284--289, Sep 2004.

\bibitem{Loefberg2009}
J.~L\"{o}fberg, ``Pre- and post-processing sum-of-squares programs in
  practice,'' {\em {IEEE} {T}rans. {A}utom. {C}ontrol}, vol.~54,
  pp.~1007--1011, may 2009.

\bibitem{Isidori:1995}
A.~Isidori, {\em Nonlinear Control Systems}.
\newblock Communications and Control Engineering, Springer, 1995.

\end{thebibliography}

\end{document}